\documentclass[12pt]{article}
\usepackage{hyperref}
\usepackage{eufrak}

\usepackage{amsfonts}
\usepackage{amsmath}
\usepackage{amssymb}
\usepackage{amsthm}

\newtheorem{thm}{Theorem}[section]
\newtheorem{lemma}[thm]{Lemma}
\newtheorem{prop}[thm]{Proposition}
\newtheorem{cor}[thm]{Corollary}
\newtheorem{defin}[thm]{Definition}
\newtheorem{rem}[thm]{Remark}
\newtheorem{exam}[thm]{Example}

\newtheorem{conjecture}[thm]{Conjecture}

\numberwithin{equation}{section}

\newcommand{\R}{{\mathbb{R}}}

\newcommand{\T}{{\mathbb{T}}}
\newcommand{\Z}{{\mathbb{Z}}}
\newcommand{\N}{{\mathbb{N}}}
\newcommand{\C}{{\mathbb{C}}}

\newcommand{\cA}{{\mathcal{A}}}

\newcommand{\cL}{{\mathcal{L}}}

\newcommand{\cX}{{\mathcal{X}}}

\def\id{{1\hskip-2.5pt{\rm l}}}

\newcommand{\tr}{{\rm tr}}
\DeclareMathOperator{\sgrad}{sgrad}

\newcommand{\til}[1]{\widetilde{#1}}

\newcommand{\End}{\textup{End}}

\newcommand{\om}{\omega}
\newcommand{\Om}{\Omega}

\newcommand{\mZ}{{\mathbb Z}}

\newcommand{\bigo}{\mathcal{O}}
\newcommand{\op}{{\textrm{op}}}

\usepackage{color}

\begin{document}

\title{Almost representations of algebras and quantization}

\author{Louis Ioos$^1$, David Kazhdan and Leonid Polterovich$^2$}
\date{}
\maketitle
\newcommand{\Addresses}{{
  \bigskip
  \footnotesize

  \textsc{School of Mathematical Sciences, Tel Aviv University, Ramat Aviv, Tel Aviv 69978
Israel}\par\nopagebreak
  \textit{E-mail address}: \texttt{louisioos@mail.tau.ac.il}
	
	\medskip
	
	\textsc{Einstein Institute of Mathematics, Hebrew University, Givat Ram, Jerusalem, 91904, Israel}\par\nopagebreak
  \textit{E-mail address}: \texttt{kazhdan@math.huji.ac.il}
	\medskip
	
	  \textsc{School of Mathematical Sciences, Tel Aviv University, Ramat Aviv, Tel Aviv 69978
Israel}\par\nopagebreak
  \textit{E-mail address}: \texttt{polterov@tauex.tau.ac.il}

}}

\footnotetext[1]{Partially supported by the European Research Council Starting grant 757585}
\footnotetext[2]{Partially supported by the Israel Science Foundation
grant 1102/20}

%%%%%%%%%
%%%%%%%%%

\begin{abstract}
We introduce  the notion of almost representations of Lie algebras and quantum tori, and establish an Ulam-stability type phenomenon: every irreducible almost representation is close to a genuine irreducible representation. As an application, we prove that geometric quantizations of the two-dimensional sphere and the two-dimensional torus are conjugate in the semi-classical limit up to a small error.
\end{abstract}

\tableofcontents

\section{Introduction and main results}\label{sec-intro}
The goal of this article is twofold.  The first objective is to establish an
Ulam-stability type phenomenon for   almost  representations of  algebras such as compact
Lie algebras and quantum tori: we show that under certain assumptions, every irreducible almost representation of such an algebra is close to a genuine irreducible representation.
In the case of Lie algebras, we present two versions
of an Ulam-stability, associated to two different
notions of an almost representation.

\begin{rem}{\rm
While a similar problem has been studied for representation of  groups \cite{Grove, K, DGLT}, to the best of our knowledge  it has not
been addressed in the framework of representations of algebras.}
\end{rem}

\noindent
Our second goal is an application of the results on stability to geometric quantization.

\subsection{Almost representations}
Let us pass to a more detailed discussion of almost representations in the three cases we consider.
\newline

\medskip
\noindent
{\sc First case:} For a finite-dimensional
Hilbert space $H$, write $\|\cdot\|_{op}$ for the operator norm on the space $\mathfrak{su}(H)$ of skew-Hermitian operators
acting on $H$.
Recall that the Lie algebra $\mathfrak{su}(2)$
has real dimension $3$,
and admits a basis $L_1,\,L_2,\,L_3\in\mathfrak{su}(2)$
satisfying the commutation relations
\begin{equation} \label{eq-gener}
[L_j,L_{j+1}] =L_{j+2}\quad\text{for all}\quad
j\in\Z/3\Z\;.
\end{equation}
An \emph{irreducible representation} is a linear map
$\rho:\mathfrak{su}(2)\to\mathfrak{su}(H)$
preserving the commutation relations and such that
the triple of skew-Hermitian operators
$X_j:=\rho(L_j),\,j\in\Z/3\Z$,
does not preserve any proper subspace of $H$.
As well known in such a case, writing $n:=\dim H$ for the complex
dimension of $H$, we have
\begin{equation}\label{Casimirfla}
X_1^2+X_2^2+X_3^2=-\frac{n^2-1}{4}\id\,.
\end{equation}
Our first main result is as follows.

\begin{thm}\label{mainthm-1} $\;$
%For every $c \in \R$ and $r >0 $,
%there exist $k_0 \in \N$ and $C >0$
%such that the following holds.
Fix $r>0$, and for every $k\in \N$ and $c \in \R$,
consider the following assumptions on a
finite-dimensional Hilbert space $H$ and
a triple of operators $x_j\in\mathfrak{su}(H), j\in \mZ /3\mZ$:
% and $k\in \N$ with $k\geq k_0$
\begin{itemize}
\item [$(\textup{R}1)$] \quad\quad\quad\quad\quad\,
$\left\|\,x_1^2 + x_2^2 + x_3^2 + \left(\frac{k^2}{4}+\frac{kc}{2}\right)
\id\,\right\|_{op}\leq r\,;$
\item [$(\textup{R}2)$] \quad\quad\quad
$\big\|\,[x_j,x_{j+1}] -x_{j+2}\,\big\|_{op}
\leq r/k\quad\text{for all}\quad
j\in\Z/3\Z\,;$
\item [$(\textup{R}3)$] \quad\quad\quad\quad\quad\quad\quad\quad\quad~
$\dim H <2(k+c)\;.$
\end{itemize}
Then the following holds:
\begin{enumerate}
\item For any $c \notin \Z$, there exists $k_0\in\N$ such that
the system of assumptions $(R1)$ and $(R2)$
cannot be fullfilled for $k\geq k_0$.
\item For any $c\in\Z$, there exists $k_0\in\N$ and
$C_1,\,C_2>0$,
such that
for all $k\geq k_0$,
for any finite
dimensional
$H$ and triple of operators $x_j\in\mathfrak{su}(H), j\in \mZ /3\mZ$,
satisfying $(R1),(R2)$ and $(R3)$, one has
\begin{equation}\label{dimH=k+c}
\dim H = k+c\,.
\end{equation}
Furthermore,
\begin{equation}\label{mainfla0}
k/2-C_1\leq\|x_j\|_{op}\leq k/2+C_1\quad\text{for all}\quad j\in\Z/3\Z\,,
\end{equation}
and there exists an irreducible representation
$\rho:\mathfrak{su}(2)\to\mathfrak{su}(H)$ satisfying
\begin{equation}\label{mainfla}
\|x_j-\rho (L_j)\|_{op}\leq C_2\quad\text{for all}\quad j\in\Z/3\Z\,.
\end{equation}
\end{enumerate}
\end{thm}

The proof of this theorem is given in Section \ref{sec-arpr}.
Let us point out that only the inequalities $(R1)$ and $(R2)$ are needed to
establish \eqref{mainfla0}.
Note also that for genuine irreducible representations of
$\mathfrak{su}(2)$, assumption $(R1)$ holds with $r=1/4$ by
\eqref{Casimirfla}, while $(R2)$ is valid for any $r>0$ by \eqref{eq-gener}.

\begin{rem}{\rm
Theorem \ref{mainthm-1} shows in particular
that for $k\in\N$ big enough,
any triple of operators $x_1,\,x_2,\,x_3\in\mathfrak{su}(H)$
with $\dim H<2(k+c)$ satisfying (R1) and (R2) acts irreducibly,
i.e., has no common proper invariant subspace $V\subset H$.
%In fact, if $V\subset H$ is a common invariant subspace,
%then so is its orthogonal complement $V^\bot\subset H$,
%and conclusion (III) of the theorem implies
%that
%$\dim V = \dim V^\bot=k+c$ whenever $V\neq H$ or $\{0\}$, which
%contradicts $\dim H<2(k+c)$. Note
Conversely,
the direct sum of two
$k$-dimensional irreducible representations satisfy the assumptions
(R1) and (R2) for $c=0$, so that the assumption
(R3) is optimal in order to get irreducible representations.
For a related discussion on almost irreducibility, see
the paragraph after Remark \ref{rem-compar}.}
\end{rem}

\begin{conjecture}
\label{conj-other} The analogue of Theorem \ref{mainthm-1} holds
for all real compact Lie algebras,
with the appropriate relation in the left hand side of (R1)
given by the Casimir element.
\end{conjecture}

Our proof of Theorem \ref{mainthm-1} uses the explicit
description of representations of $\mathfrak{su}(2)$ for all
$k\in\N$, and new ideas
are needed in order to find a uniform proof
for all compact Lie algebras in this setting.
\newline

\medskip\noindent {\sc Second case:}
The next result is a counterpart of Theorem \ref{mainthm-1}
for \emph{two-dimensional quantum tori}. Following
\cite{HB80,Rie}, recall that {\it the quantum torus}
$\cA_\theta$ is a $C^*$-algebra over $\C$ depending on a
parameter $\theta \in S^1= \R/\Z$. It is generated by two elements
$W_1,W_2\in\cA_\theta$ with relations
\begin{equation}
W_1^*W_1=W_2^*W_2=\id\quad\quad\text{and}
\quad\quad W_1W_2 = e^{2\pi i\theta}W_2W_1\,.
\end{equation}
A \emph{$*$-representation}
$\rho:\cA_\theta\to\End(H)$ on a Hilbert space $H$
is then determined by the data
of two unitary operators
$X_1:=\rho(W_1),\,X_2:=\rho(W_2)\in\End(H)$ satisfying
$X_1X_2 = e^{2\pi i\theta}X_2X_1$.
The $C^*$-algebra $\cA_\theta$ admits
an irreducible finite-dimensional $*$-representation
whenever $e^{2\pi i\theta}$ is an $n$-th prime root of unity,
and in that case we have $\dim H = n$.

\begin{thm}\label{mainthmT2}
Fix $r>0$, and for any $c \in \R$ and $k\in \N$,
consider the following assumptions on a
finite-dimensional Hilbert space $H$ and
a pair of operators
$x_1,\,x_2\in\End(H)$:
\begin{itemize}
\item [$(\textup{R}1)$] ~\quad\quad\quad\quad $\big\|x_jx_j^*-\id\big\|_{op}
\leq r/k^3\quad\text{for all}\quad
j=1,\,2\,;$
\item [$(\textup{R}2)$] \quad\quad\quad\quad\quad\quad $\left\|\,x_1x_2-e^{2i\pi/(k+c)}x_2x_1
\,\right\|_{op}\leq r/k^3\,;$
\item [$(\textup{R}3)$] \quad\quad\quad\quad\quad\quad\quad\quad\quad~
$\dim H <2(k+c)\;.$
\end{itemize}
Then the following holds:
\begin{enumerate}
\item For any $c \notin \Z$, there exists $k_0\in\N$ such that
the system of assumptions $(R1),\,(R2)$ and $(R3)$
cannot be fullfilled for $k\geq k_0$.
\item For any $c\in\Z$, there exists $k_0\in\N$ and
$C>0$
such that
for all $k\geq k_0$,
for any finite
dimensional
$H$ and a pair of operators
$x_1,\,x_2\in\End(H)$,
satisfying $(R1),\,(R2)$ and $(R3)$, one has
\begin{equation}\label{dimH=k+cT2}
\dim H = k+c\,.
\end{equation}
Furthermore, there exists a $*$-representation $\rho: \cA_\theta \to \text{End}(H)$
with $\theta = 1/(k+c)$ such that
\begin{equation}\label{UVdef}
\|x_j-\rho(W_j)\|_{op}\leq C/k^{3/2}\quad\text{for all}\quad
j=1,\,2\,.
\end{equation}
\end{enumerate}
\end{thm}
The proof follows the same strategy as the proof of Theorem \ref{mainthm-1}, see
Section \ref{QT2sec} below.
\newline

\medskip
\noindent
{\sc Third case:} In Section \ref{sec-gen}, we consider
another notion of  an irreducible almost representation $t:\mathfrak{g} \to \mathfrak{su}(H)$
of a compact Lie algebra $\mathfrak{g}$. It involves the Casimir element in the
\emph{adjoint representation}. As explained in Remark \ref{rem-compar},
this definition is weaker than the previous notion of an almost representation in the
case of the Lie algebra $\mathfrak{su}(2)$, which is necessary for applications to geometric quantization of the two-sphere.

Take any orthonormal basis $e_1,\dots,e_n$ in $\mathfrak{g}$ with respect to the Killing form. We define, in the context of almost representations, a counterpart of the Casimir element in the adjoint representation,
called {\it almost-Casimir}, by
\begin{equation}
\label{aCdef}
\begin{split}
\Gamma:\mathfrak{su}(H) &\longrightarrow\mathfrak{su}(H)\\
\sigma &\longmapsto - \sum_{i=1}^n [[\sigma,t(e_i)],t(e_i)]\;.
\end{split}
\end{equation}
We define almost representations as linear maps
$t:\mathfrak{g} \to \mathfrak{su}(H)$
which satisfy approximate commutation relations,
and such that $\Gamma$ is invertible. Theorem \ref{thm-main-1}
below provides an upper bound for the distance between such a $t$
and a  genuine irreducible representation of $\mathfrak{g}$
in terms of the operator norm of the inverse of $\Gamma$.
Here we adapt a Newton-type method  as in \cite{K}.

\subsection{Preliminaries on quantization}

Geometric quantization is a mathematical recipe behind the quantum-classical correspondence, a fundamental physical principle stating that quantum mechanics contains classical mechanics in the limiting regime when the Planck constant $\hbar$ tends to zero. In Section \ref{quantsec}, we apply
Theorems \ref{mainthm-1} and \ref{mainthmT2} to show that all geometric quantizations of the two-dimensional sphere and the torus, satisfying the axioms of Definition \ref{quantdef} below,
are conjugate to each other up to an error of order $\bigo(\hbar)$.

When the classical phase space is represented by
a closed (i.e., compact without boundary) symplectic manifold
$(M,\omega)$, geometric quantization is a linear correspondence
$f \mapsto T_\hbar(f)$
between classical observables, i.e., real functions
$f\in C^\infty(M)$ on the phase space $M$, and quantum observables, i.e.,
Hermitian operators $T_\hbar(f)\in\cL(H_\hbar)$
on a complex Hilbert space $H_\hbar$. This
correspondence is assumed to respect, in the leading order as
the \emph{Planck constant} $\hbar$ tends to $0$,
a number of basic operations.

Write $\{\cdot,\cdot\}$ for the \emph{Poisson bracket}
of $(M,\om)$, defined on any $f,\,g\in C^\infty(M)$ by
$$\{f,g\}:=-\om(\text{sgrad} f,\text{sgrad} g)\,,$$ where
for any $f\in C^\infty(M)$,
the associated \emph{Hamiltonian vector field}
$\text{sgrad} f$ over $M$ is defined by $$\om(\cdot\,,\text{sgrad} f)=df\,.$$
For any complex valued function $f\in C^\infty(M,\C)$,
we write $\|f\|_\infty:=\max_{x\in M}\,|f|$ for its uniform norm.

\begin{defin}\label{quantdef}
{\rm
Given a sequence $\{H_k\}_{k\in\N}$ of finite-dimensional complex
Hilbert spaces, an associated \emph{geometric quantization}
is a collection of $\R$-linear maps
$\{T_k:C^\infty(M)\to\cL(H_k)\}_{k\in\N}$ with $T_k(1)=\id$ for
all $k\in\N$,
satisfying the following axioms
as $k\to+\infty$,
\begin{itemize}
\item[{(P1)}] \quad\quad\quad\quad\quad\quad\quad~
$\|T_k(f)\|_{\op} = \|f\|_\infty + \bigo(1/k)\;;$
\item[{(P2)}] \quad\quad\quad\quad\quad~ $[T_k(f),T_k(g)] =\frac{i}{k}
T_k\left(\{f,g\}\right)+\bigo(1/k^2)\;;$
\item[{(P3)}] \quad\quad $T_k(f)T_k(g) = T_k\left(fg+\frac{1}{k}
C_1(f,g)+\frac{1}{k^2}C_2(f,g)\right) + \bigo(1/k^3)\;.$
\end{itemize}
In axiom (P3),
we extend $\{T_k:C^\infty(S^2,\C) \to
\text{End}(H_k)\}_{k\in\N}$ by $\C$-linearity, and the maps
$C_1,\,C_2: C^\infty(S^2) \times C^\infty(S^2)
\to C^\infty(S^2,\C)$ are bi-differential operators. The
remainders are taken with respect to the operator norm,
uniformly in the $C^N$-norms of
$f,\,g\in C^\infty(S^2)$ for some $N\in\N$.}
\end{defin}

In Definition \ref{quantdef}, the integer $k\in\N$ represents
a \emph{quantum number}, and should be thought as inversely
proportional to the Planck constant. Then the limit  $k\to+\infty$ describes the so-called \emph{semi-classical
limit}, when the scale gets so large that we recover the laws
of classical mechanics from those of quantum mechanics. In particular, the
axiom (P2) is the celebrated \emph{Dirac condition},
relating the Poisson bracket on classical observables
to the commutator bracket on quantum observables.

\noindent\begin{exam}\label{exam-BT-desc}{\rm
In the case $M$ admits a complex structure $J\in\End(TM)$
compatible with $\om$ and the De Rham cohomology class
$[\om]/2\pi$ is integral,
the existence of geometric quantizations
was established by
Bordemann, Meinrenken and Schlichenmaier \cite{BMS}, using the
theory of Boutet de Monvel and Guillemin \cite{BdMG81}.
Their construction is called
\emph{Berezin-Toeplitz quantization}, and goes as follows.
Let $L$ be a holomorphic Hermitian line bundle with
Chern curvature equalt to $-i\om$, and for any $k\in\N$, write
$L^k$ for the $k$-th tensor power of $L$.
We define the
Hilbert space $H_k$ as the subspace of all global holomorphic
sections of $L^k$ inside the Hilbert space $L_2(M,L^k)$
of $L_2$-sections of $L^k$.
With this language, the \emph{Toeplitz operators}
$T_k(f)\in\cL(H_k)$
act by multiplication by $f\in C^\infty(M)$
composed with the orthogonal projection to $H_k$ inside
$L_2(M,L^k)$.
In the case of the sphere $S^2$ and the torus $\T^2$,
by an appropriate shift of the parameter $k\in\N$,
this construction actually produces a discrete family of
geometric quantizations  depending on
$m\in\N$ and satisfying $\dim H_k=k+m$.

While the construction given above is rather straightforward,
verification of the axioms of Definition \ref{quantdef}
is highly non-trivial.
For comprehensive introductions to Berezin-Toeplitz quantization,
see for instance \cite{LF,MM,Sc}.

The Berezin-Toeplitz quantizations associated to
two distinct complex structures are essentially different, so that even
for the simplest symplectic manifolds $(M,\om)$,
this construction produces a large variety of examples.
As shown by Ma and Marinescu \cite{MM12}, Xu \cite{Xu}
and Charles \cite{Csub},
for such quantizations, the bi-differential operator $C_1(f,g)$
is proportional to the Hermitian product of the Hamiltonian
vector fields of $f$ and $g$, while the bi-differential $C_2(f,g)$
involves the Ricci curvature.}
\end{exam}

A different, albeit related,
mathematical model of quantization is
{\it deformation quantization} \cite{BFFLS}, which is an
$\hbar$-linear associative algebra on the space
$C^\infty(M)[[\hbar]]$ such that for all
$f,\,g \in C^\infty(M)$,
\begin{equation}\label{DQintro}
f*g= fg+\hbar\,C_1(f,g) + \hbar^2\,C_2(f,g) + \cdots\;,
\end{equation}
with $C_1(f,g) - C_1(g,f) =i\{f,g\}$.
Here the Planck constant $\hbar$ plays the role of a formal deformation parameter, and the operation
\eqref{DQintro} is called a {\it star-product}.
In Section \ref{trsec}, we consider geometric quantizations
satisfying an extension of axiom (P3), given in Definition
\ref{totexp}, to an asymptotic
expansion up to $\bigo(1/k^m)$ for any $m\in\N$.
This defines a star product via the formal
relation
$T_\hbar(f)T_\hbar(g)=
T_\hbar(f*g)$ with $\hbar=1/k$.
In particular, the Berezin-Toeplitz quantizations
described above satisfy this extension of axiom (P3), and
thus induce a star-product \cite{BMS, Sch, Gu} over $(M,\om)$.
While deformation quantizations of closed symplectic manifolds are
completely classified up to \emph{star equivalence}
given by formula \eqref{stareq} below, the
classification of geometric quantizations up to conjugation and an error
of order $\bigo(1/k^m)$ with given $m\in\N$
is not yet
understood. The study of this classification is the main
motivation of this paper.

\subsection{Applications to quantization} The second main result of this paper is as follows.

\begin{thm}\label{mainthquant-1}
Asssume that $(M,\om)=S^2$ or $\T^2$ endowed with
the standard area form of volume $2\pi$.
Let $\{T_k: C^\infty(M) \to \cL(H_k)\}_{k\in\N}$
be a geometric quantization, and assume that
\begin{equation}\label{eq-dim-1}
\limsup_{k\to+\infty}\dim H_k/k<2\;.
\end{equation}
Then there exists an integer $m\in\Z$ such that
for all $k\in\N$ big enough, we have
\begin{equation}\label{dimquant}
\dim H_k=k+m\,.
\end{equation}
Furthermore, if $\{Q_k: C^\infty(M) \to \cL(H_k)\}_{k\in\N}$
is another geometric
quantization associated to the same sequence
of Hilbert spaces,
there exists a sequence of unitary operators $\{U_k:H_k\to H_k\}_{k\in\N}$
such that for any $f\in C^\infty(M)$, there exists $C>0$ such that
for any $k\in\N$, we have
\begin{equation}\label{quanteq}
\|U_k^{-1}Q_k(f)U_k-T_k(f)\|_{op}\leq C/k\,.
\end{equation}
\end{thm}

From a physical point of view,
property \eqref{quanteq} is the statement
that two different quantizations
of the same classical phase space reproduce the same physics
at the semi-classical limit, when $k\to+\infty$.
Two geometric quantizations
satisfying \eqref{quanteq} are called
{\it semi-classically equivalent}.
Note on the other hand that for any $m\in\N$,
a geometric quantization
satisfying \eqref{dimquant} can be realized
through the construction of Example \ref{exam-BT-desc}.
We thus get the following corollary.

\medskip\noindent
\begin{cor}\label{cor-BT}
Under the dimension assumption \eqref{eq-dim-1},
every geometric quantization of
the sphere or the torus
is semi-classically
equivalent to a Berezin-Toeplitz quantization
of Example \ref{exam-BT-desc}.
\end{cor}

The semi-classical equivalence of the Berezin-Toeplitz quantization
and the Kostant-Souriau quantization was first shown by Schlichenmaier in \cite{Sch}. For Berezin-Toeplitz quantizations
associated with a different compatible complex structure,
Corollary \ref{cor-BT} follows from the
work of Charles \cite{Cha07}.
These results were established for more general symplectic
manifolds than the sphere and the torus, which
leads to the following
conjecture.

\begin{conjecture}
\label{question-other2} Two geometric
quantizations of a closed symplectic manifold $(M,\om)$
associated with sequences of Hilbert spaces of the same
dimension are semi-classically equivalent.
\end{conjecture}

In particular, it is not completely clear
to what extent the dimension assumption \eqref{eq-dim-1}
can be relaxed.
An affirmative answer to Conjecture \ref{conj-other} should yield
an affirmative answer to Conjecture \ref{question-other2}
in the case of
coadjoint orbits of general compact Lie groups,
at least with the appropriate assumption on the dimension.

\begin{rem}{\rm The dimension assumption
\eqref{eq-dim-1} of Theorem \ref{mainthquant-1}
is natural from a physical point of view. In fact,
equation \eqref{eq-dim-2} follows from an additional
{\it trace axiom} for geometric quantizations, which
is satisfied for Berezin-Toeplitz quantizations and which
we discuss in Section \ref{trsec}.
For geometric quantizations of closed symplectic manifolds
$(M,\omega)$ with $\dim M=2d$, this trace axiom implies
the following estimate as $k\to+\infty$,
\begin{equation} \label{eq-dim}
\dim H_k = \left(\frac{k}{2\pi}\right)^d
\text{Vol}(M,\omega) + \bigo(k^{d-1})\;.
\end{equation}
This reflects the physical principle that $\dim H_k$
approximately equals the maximal number of pair-wise disjoint
quantum cells, i.e. cubes of volume $(2\pi\hbar)^d$,
inside the classical phase space $(M,\omega)$.
When $\dim M=2$, formula \eqref{eq-dim} reads
\begin{equation} \label{eq-dim-2}
\dim H_k = k + \bigo(1)\;,
\end{equation}
so that the dimension assumption \eqref{eq-dim-1}
holds in particular for geometric quantizations
of $M=S^2$ or $\T^2$ satisfying
the trace axiom.
}
\end{rem}

In Theorem \ref{trconj}, we
consider geometric quantization satisfying a \emph{trace
axiom}, given in Definition \ref{trdef}, and we
express the asymptotics of the trace in terms
of the bi-differential operator $C_2$ of axiom (P3).
In Corollary \ref{thm-trace}, we show that this
implies the equality of the usual trace with
the \emph{canonical trace}
of the induced \emph{star product}
up to $\bigo(1/k)$, defined by formula \eqref{cantrdef} below.
Finally, in Theorem \ref{3/2eqth}, we show how
Theorem \ref{mainthmT2}
can be applied to get an extension of Theorem
\ref{mainthquant-1} for quantizations of the torus $\T^2$
which are invariant by translation,
making a first step towards
the classification of geometric quantizations up to
order $\bigo(1/k^m)$ with $m>1$.

\section{Proofs for $\mathfrak{su}(2)$ and quantum torus}
\label{sec-arpr}

Let $H$ be a Hilbert space
of complex dimension $\dim H=n\in\N$, and recall
that an operator $A\in\End(H)$ is called \emph{normal}
if it can be diagonalized in an orthonormal basis.
The following Lemma on the existence of quasimodes will
be a basic tool in this Section.

\begin{lemma}\label{quasimodo}
Let $A\in\End(H)$ be normal, and assume that
$v,\,w\in H$, $v\neq 0$, and
$\alpha\in\C$ satisfy
\begin{equation}\label{quasimododef}
Av=\alpha v+w\,.
\end{equation}
Then there exists $\lambda\in\textup{Spec}(A)$ satisfying
\begin{equation}\label{quasimodoev}
|\lambda-\alpha|\leq\frac{\|w\|}{\|v\|}\,.
\end{equation}
Furthermore, for any $\delta>0$, let
$V_\delta\subset H$ be the direct sum to the eigenspaces of
eigenvalues $\eta\in\textup{Spec}(A)$ satisfying
$|\eta-\alpha|<\delta$.
Then
there exists $e\in V_\delta$ with $\|e\|=1$
such that
\begin{equation}\label{quasimodofla}
\Big\|\,v-\,\|v\|e\,\Big\|\leq 2 \frac{\|w\|}{\delta}\,.
\end{equation}
\end{lemma}
\begin{proof}
Let $\{e_j\}_{j=1}^n$ be an orthonormal basis of $H$
diagonalizing $A$,
with complex eigenvalues $\{\lambda_j\}_{j=1}^n$.
Consider $v,\,w\in H$ and $\alpha\in\C$ satisfying formula
\eqref{quasimododef}. Then we have
\begin{equation}
\begin{split}
\|w\|=\|(A-\alpha)v\|\geq\min_{1\leq j\leq n}|\lambda_j-\alpha|
\,\|v\|\,,
\end{split}
\end{equation}
which implies that there exists $1\leq m\leq n$ such that
$\lambda_m$ satisfies formula \eqref{quasimodoev}.
%
%Write $v=\sum_{j=1}^n v_je_j$,
%so that as $Av=r$ we get
%\begin{equation}
%|r|^2=\sum_{\lambda_j}^2|v_j|^2
%\geq\min_{1\leq j\leq n}\lambda_j^2|v|^2\,,
%\end{equation}
%which implies the first part of \eqref{quasimodofla}.

Fix now $\delta>0$.
Then formula \eqref{quasimododef} implies
\begin{equation}
\begin{split}
\|w\|^2=\sum_{1\leq j\leq n}|\lambda_j-\alpha|^2\,
|\langle v,e_j\rangle|^2&\geq
\sum_{|\lambda_j-\alpha|\geq\delta}|\lambda_j-\alpha|^2\,
|\langle v,e_j\rangle|^2\\
&\geq\delta^2\,\Big\|v-
\sum_{|\lambda_m-\alpha|<\delta}\langle v,e_m\rangle\,e_m
\Big\|^2\,,
\end{split}
\end{equation}
Write $\til{e}:=\sum_{|\lambda_m-\alpha|<\delta}
\langle v,e_m\rangle\,e_m\in V_\delta$.
Then this implies in particular that
\begin{equation}
\|w\|\geq\delta\|v-\til{e}\|\geq\delta\,\Big|\,\|v\|-\|\til{e}\|\,\Big|\,.
\end{equation}
Taking $e:=\til{e}/\|\til{e}\|$, we then get
\begin{equation}
\begin{split}
\Big\|\,v-\|v\|\,e\,\Big\|&\leq
\Big\|\,v-\til{e}\,\Big\|+\Big|\,\|v\|-\|\til{e}\|\,\Big|\\
&\leq 2\frac{\|w\|}{\delta}\,.
\end{split}
\end{equation}
This proves the result.
\end{proof}

\subsection{Case of the Lie algebra $\mathfrak{su}(2)$}
\label{su2sec}

%Let $\tau:  \mathfrak{su}(2) \to \mathfrak{su}(H_k)$ be
%an irreducible representation.
A triple
of skew-Hermitian operators $X_1,\,X_2,\,X_3\in\mathfrak{su}(H)$
is said to \emph{generate an irreducible representation of
$\mathfrak{su}(2)$} if they satisfy the commutation relations
\eqref{eq-gener} and do not preserve any proper subspace of $H$.
From the basic representation theory of
$\mathfrak{su}(2)$, this is equivalent with the fact that
\begin{equation}\label{irrepsu2}
\begin{split}
X_1^2 + X_2^2 + X_3^2 &=-\left(\frac{n^2-1}{4}\right)\id
\quad\quad\text{and}\\
\quad[X_j,X_{j+1}]&=X_{j+2}\quad\quad~\text{for all}\quad
j\in\Z/3\Z\;.
\end{split}
\end{equation}
Set the \emph{ladder operators} to be
\begin{equation}\label{irreplad}
Y_{\pm}:=\pm iX_1+X_2\in\End(H)\,.
\end{equation}
Then there exists an orthonormal basis $\{e_j\}_{j=1}^n$ of
$H$ such that for all $m\in\N$ with $0\leq m\leq n-1$, we have
\begin{equation}\label{irrepX3diag}
\begin{split}
X_3\,e_{n-m}&=i\left(\frac{n-1}{2}-m\right)\,e_{n-m}\,,\\
Y_\pm\,e_{n-m}&=\mp
\sqrt{\frac{n^2-1}{4}-
\left(\frac{n-1}{2}-m\right)^2\mp\left(\frac{n-1}{2}-m\right)}
\,e_{n-m\pm 1}
\end{split}
\end{equation}
Note that the $\mp$ sign in front of the square root in the
the second line of \eqref{irrepX3diag} is a matter of
convention, as one can pass to the opposite sign by
a change of orthonormal basis $e_j\mapsto(-1)^je_j$ for all
$1\leq j\leq n$.

Conversely, if we have operators $X_3,\,Y_+,\,Y_-\in\End(H)$
satisfying \eqref{irrepX3diag} in
an orthonormal basis, then setting $X_1:=i(Y_--Y_+)/2$
and $X_2:=(Y_-+Y_+)/2$, we get three operators
$X_1,\,X_2,\,X_3\in\mathfrak{su}(2)$ generating an irreducible
representation of $\mathfrak{su}(2)$ on $H$.

Let us now compare some basic
consequences of the axioms (R1) and (R2) of Theorem \ref{mainthm-1}
with the basic theory of representations of $\mathfrak{su}(2)$
described at the beginning of the Section.
For any $k\in\N$,
introduce the ladder operators
\begin{equation}\label{ypmdef}
y_{\pm}:=\pm ix_1+x_2\in\End(H_k)\,,
\end{equation}
which satisfy $y_\pm^*=-y_\mp$. Then axiom (R2) translates to
\begin{equation}\label{usefulid1}
\big\|\pm i y_{\pm}-[x_3,y_{\pm}]\,\big\|_{op}=\bigo(1/k)\,.
\end{equation}
On the other hand, one has
\begin{equation}\label{eq-cas-3}
\begin{split}
&y_+y_- = x_1^2+x_2^2 +i[x_1,x_2]\;,\\
&y_-y_+= x_1^2+x_2^2 - i[x_1,x_2]\;,
\end{split}
\end{equation}
so that axioms (R1) and (R2) imply
\begin{equation}\label{usefulid2}
\begin{split}
&\left\| y_+y_-+\frac{(k+c)^2}{4}\id +x_3^2 -ix_3\right\|_{op}= \bigo(1)\,,\\
&\left\| y_-y_++\frac{(k+c)^2}{4}\id +x_3^2 +ix_3\right\|_{op}= \bigo(1)\,.
\end{split}
\end{equation}

\medskip\noindent{\bf Proof of Theorem \ref{mainthm-1}.}
Let us fix $c\in\R$ and $r>0$, and consider all
triples of operators $x_1, x_2, x_3\mathfrak{su}(H), j\in \mZ /3\mZ$
satisfying $(R1)$ and $(R2)$ for some finite-dimensional
Hilbert space $H$.
All the estimates in the proof are with respect to the
Hilbert norm as $k\to+\infty$.
%,
%and only depend otherwise on $c\in\R$ and $r>0$ in the
%statement of the Theorem.

The proof will be divided into 3 steps. In Step 1, we use the ladder operators \eqref{ypmdef} to construct
$k+c$ distinct eigenvectors of $x_3\in\mathfrak{su}(H)$ with
distinct eigenvalues, showing the first statement of Theorem
\ref{mainthm-1} and the inequality \eqref{mainfla0}.
In Step 2, we show that, if these eigenvectors do not
generate the whole Hilbert space $H$, we can repeat the construction
of Step 1
to get another set of $k+c$ distinct eigenvectors of
$x_3\in\mathfrak{su}(H)$, thus establishing formula \eqref{dimH=k+c}
on the dimension under the assumption (R3).
In Step 3, we define a representation
$\rho:\mathfrak{su}(2)\to\mathfrak{su}(H)$ via formulas
\eqref{irrepX3diag} for the basis constructed in Step 1, and
establish \eqref{mainfla}.

\medskip
{\sc Step 1:}
Let $\lambda_k\in\R$ be the
highest eigenvalue of the Hermitian endomorphism $-ix_3\in\End(H)$,
and let $e_k\in H$ with $\|e_k\|=1$ be such that
\begin{equation}\label{highev}
x_3e_{k}=i\lambda_k e_{k}\,.
\end{equation}
Using formula \eqref{usefulid1}, we get the estimate
\begin{equation}
x_3(y_+e_k)=i(\lambda_k+1)y_+e_k+\bigo(1/k)\,.
\end{equation}
Applying Lemma \ref{quasimodo} to $A=-ix_3$, $v=y_+e_k$,
$w=\bigo(1/k)$, and using the fact that
$\lambda_k\in\R$ is the highest eigenvalue of $-ix_3$,
we get the estimate $\|y_+e_k\|=\bigo(1/k)$. Using now formula
\eqref{usefulid2} and Cauchy-Schwartz inequality, this implies
\begin{equation}
\bigo(1/k^2)=\|y_+ e_k\|^2=-\langle y_-y_+ e_k,e_k\rangle=\frac{(k+c)^2}{4}
-\lambda_k^2-\lambda_k+\bigo(1)\,,
\end{equation}
which readily leads to the estimate
\begin{equation}\label{lambdak}
\lambda_k=\frac{k+c-1}{2}+\bigo(1/k)\,.
\end{equation}

The strategy of the proof is based on a recursive
estimate on the eigenvalues of $-ix_3$ using the lowering operators,
which we describe now. We will use the elementary fact that for any
$\epsilon>0$, there exists $\delta>0$ such that for all $k\in\N$ and
all $\lambda\in\R$, we have
\begin{multline}\label{fctstudy}
-\frac{k+c-1}{2}+\epsilon<\lambda<\frac{k+c}{2}-\epsilon\quad\text{implies}\\
\quad
\frac{(k+c)^2}{4}-
\lambda^2+\lambda>\delta (k+c)\,.
\end{multline}
Now if $f\in H$ with $\|f\|=1$ is an eigenvector of $-ix_3$ with eigenvalue
$\lambda\in\R$ satisfying \eqref{fctstudy}, we can use
formula \eqref{usefulid2} and Cauchy-Schwartz inequality to get
\begin{equation}\label{|ye|}
\begin{split}
\|y_-f\|^2=-\langle y_+y_- f,f\rangle&=\frac{(k+c)^2}{4}
-\lambda^2+\lambda+\bigo(1)\\
&\geq\delta k+\bigo(1)\,.
\end{split}
\end{equation}
On the other hand, formula \eqref{usefulid1} implies
\begin{equation}\label{xyv=yv-}
x_3(y_-f)=i(\lambda-1)y_-f+\bigo(1/k)\,.
\end{equation}
We can thus apply Lemma \ref{quasimodo} with $A=-ix_3$, $v=y_-f$ and
$w=\bigo(1/k)$ to get an eigenvector $\til{f}\in H$
of $-ix_3$ with associated eigenvalue
$\til{\lambda}\in\R$ satisfying the recursive estimates
\begin{equation}\label{recest}
\til{\lambda}=\lambda-1+\bigo(1/k^{3/2})\,.
\end{equation}
%and
%\begin{equation}\label{recestvec}
%\Big\|\,y_-f-\,\|y_-f\|\til{f}\,\Big\|= \bigo(1/k)\,.
%\end{equation}
We emphasize that the appearance of the exponent $3/2$ in the recursive
estimate \eqref{recest} is crucial for the rest of the proof.
%
%Let us now estimate the other eigenvalues of $-ix_3$
%by descending induction using the lowering operators.
%For the first step of the induction, first note that via formula
%\eqref{usefulid2} and Cauchy-Schwartz inequality,
%the estimate \eqref{lambdak} implies
%\begin{equation}\label{|yv|}
%\|y_-e_k\|^2=-\langle y_+y_- e_k,e_k\rangle=\frac{(k+c)^2}{4}
%-\lambda_k^2+\lambda_k+\bigo(1)=k+\bigo(1)\,,
%\end{equation}
%so that in particular $y_-e_k\in H$ does not vanish for
%$k\in\N$ big enough. On the other hand, formula
%\eqref{usefulid1} implies
%\begin{equation}
%x_3(y_-e_k)=i(\lambda_k-1)y_-e_k+\bigo(1/k)\,,
%\end{equation}
%so that applying Lemma \ref{quasimodo} to $A=-ix_3$, $v=y_-e_k$,
%$w=\bigo(1/k)$, we get an eigenvalue
%$\lambda_{k-1}\in\R$ of $-ix_3$ such that
%\begin{equation}\label{lambdak-1}
%\lambda_{k-1}=\lambda_k-1+\bigo(1/k^{3/2})\,.
%\end{equation}

Fix $\epsilon>0$ small enough in \eqref{fctstudy},
and recall the estimate \eqref{lambdak} for the highest eigenvalue
of $-ix_3$. Let us prove by induction that for all $m\in\N$ satisfying
$0\leq m<k+c$, there is a normalized eigenvector $e_{k-m}\in H$ of $-ix_3$
with associated eigenvalue $\lambda_{k-m}\in\R$ satisfying
\begin{equation}\label{lambdam}
\lambda_{k-m}=\lambda_k-m+m\,\bigo(1/k^{3/2})\,.
\end{equation}
Note that this is trivially satisfied for $m=0$.
On the other hand, we see from \eqref{lambdak} that
for any $0\leq m<k+c-1$, the right hand side of
\eqref{lambdam} satisfies \eqref{fctstudy} as soon as $k\in\N$ is big enough.
Thus if \eqref{lambdam} is satisfied for some $0\leq m<k+c-1$,
we can apply the recursive estimate \eqref{recest}
to $\lambda=\lambda_{k-m}$ to get
\begin{equation}\label{lambdak-m-1}
\lambda_{k-m-1}=\lambda_{k-m}-1+\bigo(1/k^{3/2})\,,
\end{equation}
which implies \eqref{lambdam} with $m$ replaced by $m+1$. This shows
by induction that for all $m\in\N$
satisfying $0\leq m<k+c$, we have
\begin{equation}\label{lambdamfin}
\lambda_{k-m}=\frac{k+c-1}{2}-m+\bigo(1/\sqrt{k})\,.
\end{equation}

Let us now write $\lambda_-\in\R$ for the lowest eigenvalue
of $-ix_3$. The argument leading to the estimate \eqref{lambdak}
using $y_-$ instead of $y_+$ leads to the estimate
\begin{equation}\label{lambda-}
\lambda_-=-\frac{k+c-1}{2}+\bigo(1/k)\,.
\end{equation}
%On the other hand, as $\lambda_-$ is the lowest eigenvalue,
%we must have $\lambda_-\leq\lambda_{k-m}$ for all
%$0\leq m<k+c$.
Assume without loss of generality
that $c\geq 0$. Suppose, on the contrary, that $c\notin\Z$.
Then $m=\lfloor k+c\rfloor < m+c$.
Applying \eqref{lambdamfin}, we get that $\lambda_{k-m}$ is
smaller than the lowest eigenvalue
\eqref{lambda-} for
$k\in\N$ big enough. This contradiction shows that necessarily $c\in\Z$,
which proves the first statement of Theorem \ref{mainthm-1}.
We also get that $\|x_3\|_{op} = k/2 +\bigo(1)$,
and by symmetry of the assumptions (R1) and (R2) in $j\in\Z/3\Z$,
this yields \eqref{mainfla0}.

\medskip
{\sc Step 2:}
Let us now assume that (R3) holds, in addition to (R1) and (R2).
Our goal is to establish formula \eqref{dimH=k+c} on the dimension.
By the first statement of Theorem \ref{mainthm-1}
and through the shift $k\mapsto k+c$,
we will assume without loss of generality that $c=0$. Using the
estimate \eqref{lambdamfin}, we get a set
of eigenvalues of $-ix_3$ parametrized by $m\in\N$ with
$0\leq m\leq k-1$, which are pairwise distinct for $k\in\N$
big enough. This implies that $\dim H\geq k$.

To establish formula \eqref{dimH=k+c} with $c=0$,
let us consider $k\in\N$ big enough and
assume on the contrary that $\dim H\geq k+1$.
Let $E\subset H$ be the direct sum of $1$-dimensional eigenspaces
associated with each of the eigenvalues \eqref{lambdamfin}, so
that $\dim E=k$ for $k\in\N$ big enough,
and $E$ is a proper subspace of $H$. In particular, there
exists an eigenvalue $\mu\in\R$ of $-ix_3$ admitting
an eigenvector $e_{\mu}\notin E$.
Note that $\mu$ has to lie between the
highest eigenvalue \eqref{lambdak}
and the lowest eigenvalue \eqref{lambda-}.

Furthermore, we claim that for any $\epsilon>0$, there exists
$k_0\in\N$ such for any $k\geq k_0$, there exists $0\leq m_0\leq k-1$
such that
\begin{equation}\label{lambdatilm0}
\left|\mu-\lambda_{k-m_0}\right|<2\epsilon\,.
\end{equation}
%Furthermore, if
%we write $$\til{c}:=\frac{\lambda_{k}-\til{\lambda}}{2}\,,$$
%In fact,
%Fixing $\epsilon>0$ small enough in \eqref{fctstudy} with $c=0$,
Indeed, considering \eqref{fctstudy} and \eqref{lambdamfin}
with $c=0$,
we can use repeatedly the recursive estimate
\eqref{recest} as in Step 1 to produce eigenvalues
\begin{equation}\label{lambdamfintil}
\mu_m:=\mu-m+m\bigo(1/k^{3/2})\,,
\end{equation}
for all $m\in\N$ such that $\mu_m>-\frac{k-1}{2}+\epsilon$.
%satisfying
%\begin{equation}
%\frac{k-\til{c}-1}{2}-m>-\frac{k-1}{2}+\epsilon\,.
%\end{equation}
%\todo{argument precised with the correct estimates}
Assume, on the contrary, that \eqref{lambdatilm0} is not satisfied
for some $0\leq m_0\leq k-1$. Then we can
%Assuming by contradiction that we have
%$\text{dist}(2\mu,\Z)>3\epsilon$ for $k\in\N$ big enough,
use the recursive estimate \eqref{recest} one more time
to produce an eigenvalue
$\mu_-$ satisfying $\mu_-<-\frac{k-1}{2}-\epsilon$.
Thus, $\mu_-$ is smaller than the lowest eigenvalue
\eqref{lambda-}. This contradiction proves \eqref{lambdatilm0}.

Now for any $m\in\N$ and $\theta>0$, write
\begin{equation}\label{Vmdef}
V_m(\theta):=\bigoplus_{\left|\lambda-\lambda_{k-m}\right|<\theta}\,E_\lambda\,,
\end{equation}
where $E_\lambda:=\{\,v\in H\,|\,-ix_3v=\lambda v\,\}$ for all
$\lambda\in\R$.
By assumption, there exists an eigenvector
$e_{\mu}\in H$
associated with $\mu$ which does not belong to the
eigenspace associated with $\lambda_{k-m_0}$ in $E\subset H$.
Thus inequality \eqref{lambdatilm0} implies
\begin{equation}
\dim V_{m_0}(2\epsilon)\geq 2\,.
\end{equation}
Note that for $k\in\N$ big enough, we have
either $m_0>0$ or $m_0<k-1$ (or both).
Without loss of generality, let us assume that $m_0<k-1$.

We claim that
\begin{equation}\label{Vmeps}
\dim V_m(2\epsilon+m\bigo(1/k^{3/2}))\geq 2\,,
\end{equation}
for every $m_0 \leq m \leq k-1$. The proof goes by induction in $m$.
We have already seen that \eqref{Vmeps} holds true for $m=m_0$. Assume it
is valid for some $m < k-1$. Then there exists an eigenvalue
$\til\mu_m\in\R$ of $-ix_3$ satisfying
\begin{equation}
|\lambda_{k-m}-\til\mu_m| \leq 2\epsilon+m\bigo(1/k^{3/2})\;,
\end{equation}
with associated eigenvector orthogonal to the subspace $E\subset H$.
Applying the recursive estimate \eqref{recest} and comparing
with \eqref{lambdak-m-1}, we get an eigenvalue $\til\mu_{m+1}\in\R$
of $-ix_3$, which satisfies
\begin{equation}
|\lambda_{k-m-1}-\til\mu_{m+1}|<2\epsilon+ (m+1)\bigo(1/k^{3/2})\;.
\end{equation}
Hence if $\til\mu_{m+1}\neq\lambda_{k-m-1}$, the estimate
\eqref{Vmeps} for $m+1$ holds automatically. We can thus assume
without loss of generality that $\til\mu_{m+1}=\lambda_{k-m-1}$.
In this case, the recursive estimate \eqref{recest} implies
\begin{equation}
|\lambda_{k-m} - \til\mu_{m}| = \bigo(1/k^{3/2})\,.
\end{equation}
Consider now two normalized orthogonal eigenvectors
$f_1\in E$ and $f_2 \notin E$ respectively
associated with $\lambda_{k-m}$ and $\til\mu_{m}$.
Applying the second part of Lemma \ref{quasimodo}
with $\delta=2\epsilon$ to $y_- f_1$ and $y_-f_2$,
we get vectors
\begin{equation}
\til{f}_1,\,\til{f_2}\in V_{m+1}(2\epsilon + \bigo(1/k^{3/2}))
\end{equation}
with $\|\til{f}_1\|=\|\til{f}_2\|=1$ such that for $j=1,\,2$, we have
\begin{equation}\
\Big\|\,y_-f_j-\,\|y_-f_j\|\til{f}_j\,\Big\|= \bigo(1/k)\;.
\end{equation}
On the other hand, using formula \eqref{usefulid2}
and the fact that $\langle f_1,f_2\rangle=0$,
we have
\begin{equation}
\langle y_-f_1,y_-f_2\rangle=-\langle y_+y_-f_1,f_2\rangle=\bigo(1)\,.
\end{equation}
Furthermore, formula \eqref{|ye|} shows that for $j=1,\,2$, we
have
\begin{equation}
\|y_-f_j\|^2\geq\delta k+\bigo(1)\,.
\end{equation}
This gives
\begin{equation}
\langle\til{f}_1,\til{f}_2\rangle=\frac{1}{\|y_-f_1\|}
\frac{1}{\|y_-f_2\|}\langle y_-f_1,y_-f_2\rangle+\bigo(1/k^2)=\bigo(1/k^2)\,,
\end{equation}
so that
$\til{f}_1,\,\til{f}_2
\in V_{m+1}(2\epsilon+ \bigo(1/k^{3/2}))$
are linearly independent
for $k\in\N$ big enough. This finishes the proof of claim \eqref{Vmeps}.
%
%, which completes Case 2.
%(Note that under assumptions of Step 2, we proved a stronger statement than
%claim \eqref{Vmeps}).

It follows from \eqref{Vmeps} that for all
$m_0\leq m\leq k-1$,
we have
\begin{equation}\label{Vmeps2}
\dim V_m(2\epsilon+\bigo(1/\sqrt{k}))\geq 2\,.
\end{equation}
On the other hand, if $m_0>0$, we can repeat
the same process starting with $V_{m_0}(2\epsilon)$
using $y_+$ instead of $y_-$ to get
\eqref{Vmeps2} for all
$m\in\N$ with $0\leq m\leq k-1$.

By definition \eqref{Vmdef},
the subspaces
$V_m(2\epsilon+\bigo(1/\sqrt{k}))$ are pairwise orthogonal for each $m\in\N$
as soon as $k\in\N$ is big enough, and we have
\begin{equation}
\dim\bigoplus_{0\leq m\leq k-1}
V_m(2\epsilon+\bigo(1/\sqrt{k})))\geq 2k\,.
\end{equation}
This contradicts the assumption (R3) for $c=0$, and
shows that $\dim H=k$, thus proving \eqref{dimH=k+c}.

\medskip
{\sc Step 3:}
We are now left with constructing a representation satisfying \eqref{mainfla}.
Assuming without loss of generality that $c=0$, the argument above
shows in particular
that all eigenvalues of $-ix_3$ are simple and satisfy formula
\eqref{lambdamfin} for all $m\in\N$ with $0\leq m\leq k-1$.
Using Lemma \eqref{quasimodo},
we get a normalized eigenvector $e_{m-1}\in H$ of $-ix_3$ associated with $\lambda_{m-1}$ satisfying
\begin{equation}\label{ypmapprox}
\begin{split}
\langle y_-e_m,e_{m-1}\rangle
&=\langle\|y_-e_m\|e_{m-1},e_{m-1}\rangle+\bigo(1/k)\\
&=\|y_-e_m\|+\bigo(1/k)\,.
\end{split}
\end{equation}
Starting with any eigenvector $e_k$ of
$-ix_3$ associated with $\lambda_k$, we can then
construct an orthonormal eigenbasis $\{e_j\}_{j=1}^k$ for $x_3$ associated
to the sequence of eigenvalues $\{\lambda_j\}_{j=1}^k$ and
satisfying formula \eqref{ypmapprox} for
all $1\leq m\leq k$. Let us now note that for any $\lambda\in\R$,
using in particular formula \eqref{fctstudy}, we have that
\begin{multline}
-\frac{k-1}{2}+\epsilon<\lambda<\frac{k}{2}-\epsilon
\quad\text{implies}\\
\left|\frac{d}{d\lambda}\left(\sqrt{\frac{k^2-1}{4}-
\lambda^2+\lambda}\right)\right|=\bigo(\sqrt{k})\,,
\end{multline}
Via the first line of \eqref{|ye|} and
formula \eqref{lambdamfin}, this implies for all $1\leq m\leq k$ that
\begin{equation}\label{y-est}
\|y_-e_{k-m}\|=\sqrt{\frac{k^2-1}{4}-
\left(\frac{k-1}{2}-m\right)^2+\left(\frac{k-1}{2}-m\right)
}+\bigo(1)\,.
\end{equation}
On the other hand, for all $j\neq m-1$,
using
formula \eqref{usefulid1} and Cauchy-Schwartz inequality,
we get
\begin{equation}\label{ypmapprox1'}
\begin{split}
i\langle y_-e_m,e_{j}\rangle&=\langle [x_3,y_-]e_m,e_{j}\rangle
+\bigo(1/k)\\
&=i(\lambda_{m}-\lambda_{j})\,
\langle y_-e_m,e_{j}\rangle+\bigo(1/k)\,.\\
\end{split}
\end{equation}
Now formula \eqref{lambdak} implies that
$|\lambda_{j}-\lambda_{m}-1|\geq 1/2$
as soon as $j\neq m-1$ and $k\in\N$ big enough, so that
\eqref{ypmapprox1'} implies
\begin{equation}\label{ypmapprox2}
\langle y_-e_m,e_{j}\rangle=\bigo(1/k)\quad\text{for}\quad j\neq m-1\,,
\end{equation}
and we get analogous formulas for $y_+=-y_-^*$ by definition \eqref{ypmdef}.
%, we get for
%all $1\leq j,\,m\leq k$ that
%\begin{equation}
%\begin{split}
%\langle y_+e_m,e_{j}\rangle&=\bigo(1/k)\quad\text{for}\quad j\neq m+1\,,\\
%\langle y_+e_m,e_{m+1}\rangle&=\sqrt{\frac{k^2}{4}-\frac{1}{4}-
%\left(\frac{k-1}{2}-m\right)^2-\left(\frac{k-1}{2}-m\right)
%}+\bigo(1)\,.
%\end{split}
%\end{equation}
%%%\todo{wrong formula, to correct}

In the orthonormal basis $\{e_j\}_{j=1}^k$ of $H$ constructed
above and following \eqref{irrepX3diag} and \eqref{irrepsu2},
let us now set
\begin{equation}\label{y+est}
\begin{split}
X_3\,e_{n-m}&:=i\left(\frac{k-1}{2}-m\right)\,e_{n-m}\,,\\
Y_\pm\, e_{n-m}&:=\mp\sqrt{\frac{k^2-1}{4}-
\left(\frac{k-1}{2}-m\right)^2\mp\left(\frac{k-1}{2}-m\right)}
\,e_{n-m\pm 1}\,,
\end{split}
\end{equation}
for all $0\leq m\leq k-1$. By the basic representation theory
of $\mathfrak{su}(2)$ described at the beginning of the section
and the definition \eqref{ypmdef} of $y_\pm$,
to show Theorem \ref{mainthm-1}, it suffices to show that
$\|x_3-X_3\|_{op}=\bigo(1)$ and $\|y_\pm-Y_\pm\|_{op}=\bigo(1)$.
Now formula \eqref{lambdam} implies immediately that
\begin{equation}
\|x_3-X_3\|_{op}=\bigo(1/\sqrt{k})\,.
\end{equation}
On the other hand, for all $1\leq j,\,m\leq k$, formulas
\eqref{ypmapprox}, \eqref{y-est}, \eqref{ypmapprox2}
and \eqref{y+est} yield
\begin{equation}
\begin{split}
\langle(y_\pm-Y_\pm)e_j,e_m\rangle
=\bigo(1/k)\quad&\text{for}\quad m\neq j\pm 1,,\\
\langle (y_\pm-Y_\pm)e_m,e_{m\pm 1}\rangle&=\bigo(1)\,.
\end{split}
\end{equation}
Decompose the matrix into $y_\pm-Y_\pm=A+B$, where all coefficients
of $A$ vanish except $A_{m,m\pm 1}=\bigo(1)$ for all $1\leq m\leq k$
and where $B_{jm}=\bigo(1/k)$ for all $1\leq j,\,m\leq k$. Then
we readily get $\|A\|_{op}=\bigo(1)$, while by Cauchy-Schwartz we
compute
\begin{equation}
\begin{split}
\|B\|_{op}^2&=\max_{\|v\|=1}\sum_{j=1}^k
\left|\sum_{m=1}^k B_{jm}\langle e_m,v\rangle\right|^2\\
&\leq k\max_{1\leq j\leq k}\sum_{m=1}^k|B_{jm}|^2\\
&\leq k^2\bigo(1/k^2)=\bigo(1)\,.
\end{split}
\end{equation}
%\begin{equation}
%\begin{split}
%\|y_\pm-Y_\pm\|_{op}^2&\leq k\max_{1\leq j\leq k}\sum_{m=1}^k
%|\langle\left(y_\pm-Y_\pm\right)e_j,e_m\rangle|^2\\
%&\leq k
%\langle e_m,\left(y_\pm-Y_\pm\right)e_j\rangle\\
%=\langle y_\pm e_j-e_{j-1}&,e_{j-1}\rangle
%\langle e_{j-1},y_\pm e_j-e_{j-1}\rangle\\
%&+\sum_{0\leq m\leq k,\,m\neq j-1}
%\langle y_\pm e_j-e_{j-1},e_m\rangle
%\langle e_m,y_\pm e_j-e_{j-1}\rangle\\
%&=\bigo(1)+k\,\bigo(1/k^2)=\bigo(1)\,,
%\end{split}
%\end{equation}
By the triangle inequality this gives
\begin{equation}
\|y_\pm-Y_\pm\|_{op}\leq\|A\|_{op}+\|B\|_{op}=\bigo(1)\,.
\end{equation}
We get \eqref{mainfla} for the representation defined by \eqref{y+est}
as described in the beginning of the section.
This concludes the proof of Theorem \ref{mainthm-1}.

\qed

\subsection{Case of the quantum torus}
\label{QT2sec}

A pair of unitary operators $X_1,\,X_2\in\End(H)$
\emph{generates an irreducible
representation of the quantum torus} $\cA_{1/n}$
 if it satisfies a commutation relation
\begin{equation}
X_1X_2=e^{2\pi i/n}X_2X_1\,.
\end{equation}
Diagonalizing $X_1$ in an orthonormal basis
$\{e_m\}_{m=1}^n$ of $H$,
we readily get that
$X_1,\,X_2\in\End(H)$ generate an irreducible representation
of the quantum torus if and only if there exists
$\theta_1,\,\theta_2\in \R/n\Z$
such that
\begin{equation}\label{X1X2}
\begin{split}
X_1&:=\textup{diag}\left(e^{2\pi i\frac{\theta_1}{n}},e^{2\pi i
\frac{\theta_1+1}{n}},
\cdots,e^{2\pi i\frac{\theta_1+m}{n}},\cdots,
e^{2\pi i\frac{\theta_1+n-1}{n}}\right)\,,\\
X_2 e_m&:=e^{2\pi i\frac{\theta_2}{n}}e_{m+1}
\quad\text{for all}\quad m\in\Z/k\Z\,.
\end{split}
\end{equation}
Note that in this case $X_1$ and $X_2$ have no non-trivial common invariant subspace.

\medskip\noindent{\bf Proof of Theorem \ref{mainthmT2}.}
All the estimates in the proof are with respect to the
Hilbert norm as $k\to+\infty$,
and only depend otherwise on $c\in\R$ and $r>0$ in the
statement of the Theorem.

Consider the polar decomposition $x_j=P_jU_j$, where
$U_j\in\End(H)$ is unitary and $P_j\in\End(H)$ is
positive Hermitian
for each $j=1,\,2$. Then
axiom (R1) is equivalent to $\|P_j^2-\id\|_{op}=\bigo(1/k^3)$,
which implies that $\|P_j-\id\|_{op}=\bigo(1/k^3)$. Using
the submultiplicativity of the operator
norm, we then see that the unitary parts
$U_1,\,U_2\in\End(H)$ also satisfy the axioms (R1) and (R2),
and that $\|x_j-U_j\|_{op}=\bigo(1/k^3)$ for each $j=1,\,2$.
%First note that if $x_1,\,x_2\in\End(H)$ are not diagonalizable
%for
%all $k\in\N$, then by a classical density result, there
%exist diagonalizable endomorphisms
%$\til{x}_1,\,\til{x}_2\in\End(H)$ for all $k\in\N$ such that
%$\left\|\til{x}_j-x_j\right\|_{op}=\bigo(1/k^3)$
%for every $j=1,\,2$.
%One readily checks that they also satisfy the axioms (R1) and (R2),
%and if \eqref{UVdef} holds for them, it also holds
%for $x_1,\,x_2\in\End(H)$, $k\in\N$.
We are thus reduced to the case of
$x_1,\,x_2\in\End(H)$ being unitary endomorphisms.
In particular, they are normal
and Lemma \ref{quasimodo} applies.

The proof of Theorem \ref{mainthmT2} will be divided into 3 steps,
following the structure of the proof of Theorem \ref{mainthm-1}. In Step 1, we construct a set of
eigenvectors for
$x_1\in\End(H)$ using $x_2\in\End(H)$ as a ladder operator.
Then the only notable difference with the proof of Theorem \ref{mainthm-1} is that in the case of Theorem \ref{mainthmT2},
Statement 1 and formula \eqref{dimH=k+cT2} on the dimension
depend on each other, and are established together
in Step 2.

\medskip
{\sc Step 1:}
Let $e\in H$ be an eigenvector
of $x_1\in\End(H)$, and write $\lambda_0\in\C$
for the associated eigenvalue. Then axiom (R2) implies
\begin{equation}\label{evcomm}
x_1(x_2e)=\lambda_0 e^{\frac{2\pi i}{k+c}}x_2e+\bigo(1/k^3)\,.
\end{equation}
As $x_2$ is unitary by assumption, we have $\|x_2e\|=1$,
so that Lemma \ref{quasimodo},
with $A=x_1$, $v=x_2 e$ and
$w=\bigo(1/k^3)$,
implies the existence
of a normalized eigenvector
$e_1\in H$ of $x_1$ with associated eigenvalue $\lambda_1\in\C$
satisfying
\begin{equation}\label{lambda1tor}
\lambda_1=\lambda_0 e^{\frac{2\pi i}{k+c}}+\bigo(1/k^3)\,.
\end{equation}
We then obtain by induction
eigenvalues $\lambda_m\in\C$
for all $0\leq m<k+c$ satisfying
\begin{equation}\label{lambdam-1}
\begin{split}
\lambda_m&=\lambda_0 e^{\frac{2\pi im}{k+c}}+m\,\bigo(1/k^3)\\
&=\lambda_0 e^{\frac{2\pi im}{k+c}}+\bigo(1/k^2)\,.
\end{split}
\end{equation}
As $|\lambda_m|=1$ by unitarity,
these eigenvalues are distinct for all $0\leq m<k+c$
as soon as $k\in\N$ is big enough. In particular, we have
$\dim H\geq k+c$ as soon as $k\in\N$ is big enough.

\medskip
{\sc Step 2:}
Our goal now is to establish Statement 1 and formula
\eqref{dimH=k+cT2} of Theorem \ref{mainthmT2}.
Note that if $c\notin\Z$, then
$\lambda_0\neq\lambda_0 e^{\frac{2\pi i\lfloor k+c\rfloor}{k+c}}$,
and one can apply the construction of Step 1 once more to get distinct
eigenvalues of the form \eqref{lambdam-1} for all
$0\leq m<k+c+1$, so that
in particular $\dim H\geq k+c+1$ for $k\in\N$ big enough.
Hence to show that $c\in\Z$ and $\dim H=k+c$, it suffices
to show that $\dim H<k+c+1$.

Assume by contradiction that $\dim H\geq k+c+1$.
Let $E\subset H$ be the direct sum of $1$-dimensional eigenspaces
associated with each of the eigenvalues \eqref{lambdam-1}
for all $0\leq m<k+c$, so
that $\dim E<k+c+1$,
and $E$ is a proper subspace of $H$. In particular, there
exists an eigenvalue $\til{\lambda}_0\in\C$ of $x_1$ admitting
an eigenvector $\til{e}_{0}\notin E$.
Assume first that we have
\begin{equation}\label{lambdatilm0>}
\left|\til{\lambda}_0-\lambda_{m}\right|\geq \frac{\pi}{2k}\,,
\end{equation}
for all $0\leq m<k+c$.
Then we can repeat the reasoning of Step 1
to construct by induction
eigenvalues $\til{\lambda}_m\in\C$
for all $0\leq m<k+c$ satisfying
\begin{equation}\label{tillambdam}
\til{\lambda}_m=\til{\lambda}_0 e^{\frac{2\pi im}{k+c}}+\bigo(1/k^2)\,.
\end{equation}
As $|\til{\lambda}_m|=1$ by unitarity,
they are all distinct for all $0\leq m<k+c$
as soon as $k\in\N$ is big enough,
and \eqref{lambdatilm0>} also implies that they are distinct
from the set \eqref{lambdam-1} for $k\in\N$ big enough.
This implies in particular that $\dim H\geq 2(k+c)$,
which contradicts the assumption of the Theorem.

Let us now consider the remaining case
\begin{equation}\label{lambdatilm0-1}
\left|\til{\lambda}_0-\lambda_{m_0}\right|<\frac{\pi}{2k}\,,
\end{equation}
for some $0\leq m_0<k+c$.
%%\todo{argument precised following the $S^2$ case}
For any $m\in\N$ and $\mu>0$, write
\begin{equation}\label{Vmdef-1}
V_m(\mu):=\bigoplus_{\left|\lambda-\lambda_{m}\right|<
\mu}\,E_\lambda\,,
\end{equation}
where $E_\lambda:=\{\,v\in H\,|\,x_1v=\lambda v\,\}$ for all
$\lambda\in\R$.
Now by assumption, there exists an eigenvector
$\til{e}_0\in H$ of $x_1$
associated with $\til{\lambda}_0$ which does not belong to a
$1$-dimensional eigenspace associated with $\lambda_{m_0}$.
The assumption \eqref{lambdatilm0-1}
translates to
\begin{equation}\label{Vmpieps}
\dim V_{m_0}\left(\frac{\pi}{2k}\right)\geq 2\,.
\end{equation}

We claim that
\begin{equation}\label{Vmpiepsclaim}
\dim V_{m}\left(\frac{\pi}{2k}+m\bigo(1/k^3)\right)\geq 2\,,
\end{equation}
for every $m_0\leq m\leq k+c$. The proof goes by induction
in $m$, the case $m=m_0$ being given by \eqref{Vmpieps}.
Assume that it is valid for some $m<k-1$.
Then there exist an eigenvalue $\mu_m\in\R$
of $x_1$ satisfying
\begin{equation}\label{mumT2}
\left|\lambda_{m}-\mu_m\right|<\frac{\pi}{2k}+m\bigo(1/k^3)\,,
\end{equation}
with associated eigenvector orthogonal to the subspace $E\subset H$.
Then applying the recursive process of Step 1,
we get an eigenvalue $\mu_{m+1}\in\R$ of $x_1$ satisfying
\begin{equation}\label{mum+1T2}
\left|\lambda_{m+1}-\mu_{m+1}\right|<\frac{\pi}{2k}+(m+1)\bigo(1/k^3)\,,
\end{equation}
Hence if $\lambda_{m+1}\neq\mu_{m+1}$, the claim \eqref{Vmpiepsclaim}
for $m+1$ holds automatically.
We can thus assume withou loss of generality that $\lambda_{m+1}=\mu_{m+1}$.
In this case, the recursive process of Step 1 implies
\begin{equation}\label{mumT22}
\left|\lambda_{m}-\mu_m\right|=\bigo(1/k^3)\,.
\end{equation}
Consider now two normalized orthogonal eigenvectors $f_1\in E$ and
$f_2\notin E$ respectively associated with $\lambda_m$ and $\mu_m$.
Applyting the second part of
Lemma \ref{quasimodo} with $w=\bigo(1/k^3)$ and $\delta=\pi/2k$
applied to $v=x_2 f_1$ and $x_2 f_2$ respectively,
and using the unitarity of $x_2\in\End(H)$, we get vectors
\begin{equation}
\til{f}_1,\,\til{f}_2\in V_{m+1}\left(\frac{\pi}{2k}+\bigo(1/k^3)\right)
\end{equation}
with $\|\til{f}_1\|=\|\til{f}_2\|=1$ such that
for all $j=1,\,2$, we have
\begin{equation}\label{fj}
x_2f_j=\til{f}_j+\bigo(1/k^2)\,.
\end{equation}
%
% so that $\|f_1-f_2\|=\sqrt{2}$.
%Then by Lemma \ref{quasimodo} with $A=x_1$, $v=x_2 f_j$ and
%$w=\bigo(1/k^3)$ as in \eqref{evcomm}
%we
This implies
%\todo{argument precised}
\begin{equation}
\begin{split}
\left\|\til{f}_1-\til{f_2}\right\|\geq\|f_1-f_2\|
\,(1+\bigo(1/k^2))
=\sqrt{2}\,(1+\bigo(1/k^2))\,,\\
\left\|\til{f}_1+\til{f_2}\right\|\leq
\|f_1+f_2\|\,(1+\bigo(1/k^2))
=\sqrt{2}\,(1+\bigo(1/k^2))\,,
\end{split}
\end{equation}
so that by \eqref{fj} and the unitarity of $x_2$,
the vectors
$\til{f}_1,\,\til{f}_2\in V_{m+1}\left(\frac{\pi}{2k}+\bigo(1/k^3)\right)$
are linearly independent for $k\in\N$
big enough. This finishes the proof of claim \eqref{Vmpiepsclaim}.

Now by definition \eqref{Vmdef-1},
these subspaces are pairwise orthogonal for each $m\in\N$ as soon
as $k\in\N$ is big enough, and we have
\begin{equation}
\dim\bigoplus_{0\leq m< k+c-1}
V_m\left(\frac{\pi}{2k}+m\bigo(1/k^3)\right)\geq 2(k+c)\,.
\end{equation}
This contradicts the assumption
$\dim H<2(k+c)$, and proves
that $c\in\Z$ and $\dim H=k+c$. Thus we established the first statement of the theorem
and \eqref{dimH=k+cT2}.

\medskip
{\sc Step 3:}
We are now left with constructing a $*$-representation
satisfying \eqref{UVdef}.
Via the shift $k\mapsto k+c$, it suffices to consider the
case $c=0$.
Given an
eigenvector $e_m\in H$ of $x_1$ associated with $\lambda_m\in\C$
as in formula \eqref{lambdam-1} and
applying Lemma \ref{quasimodo} with $A=x_1$, $v=x_2 e_m$ and
$w=\bigo(1/k^3)$ as in formula \eqref{fj},
we can choose the eigenvector $e_{m+1}\in H$ of $x_1$
asociated with $\lambda_{m+1}\in\C$ such that
\begin{equation}\label{lvlop1}
x_2e_{m}=e_{m+1}+\bigo(1/k^2)\,.
\end{equation}
Starting from an arbitrary eigenvector $e_0\in H_m$ of $x_1$
associated with $\lambda_0$,
we construct in this way an eigenbasis $\{e_m\}_{m=0}^{k-1}$
for $u_k$, and using Lemma \ref{quasimodo} again,
we get $\theta\in\R$ such that
\begin{equation}\label{lvlop2}
x_2e_{k-1}=e^{i\theta}\,e_0+\bigo(1/k^2)\,.
\end{equation}
Setting $f_m:=e^{-i\theta m/k}e_m$ for all $m\in\Z/k\Z$
and working in the basis $\{f_m\}_{m=0}^{k-1}$ instead,
set
\begin{equation}\label{X1X2'}
\begin{split}
X_1&:=\textup{diag}\left(\lambda_0,\lambda_0e^{2i\pi/k},
\cdots,\lambda_0e^{2i\pi m/k},\cdots,
\lambda_0e^{2i\pi(k-1)/k}\right)\,,\\
X_2 f_m&:=e^{i\theta/k}f_{m+1}
\quad\text{for all}\quad m\in\Z/k\Z\,.
\end{split}
\end{equation}
By construction, the endomorphism $x_1\in\End(H)$ is diagonal
in the same basis than $X_1$ with eigenvalues given by formula
\eqref{lambdam-1}, so that
\begin{equation}
\left\|X_1-x_1\right\|_{op}=\bigo(1/k^2)\,.
\end{equation}
Furthermore, Cauchy-Schwartz inequality together
with formulas
\eqref{lvlop1} and \eqref{lvlop2} imply
\begin{equation}
\left\|X_2-x_2\right\|_{op}=\bigo(1/k^{3/2})\,.
\end{equation}
Comparing with \eqref{X1X2}, we get \eqref{UVdef}, which completes the proof of the theorem.

\qed

\section{Almost representations of compact Lie algebras}\label{sec-gen}

In this section, we propose an alternative notion of irreducibility of almost
representations in the context of general compact Lie algebras,
and present another version of the Ulam-type statement: irreducible almost-representations
can be approximated by a genuine representation.

Let $(\mathfrak{g},\{\cdot,\cdot\})$
be a real \emph{compact} $n$-dimensional
Lie algebra. This means that it is semi-simple and
that its Killing form $\langle\cdot,\cdot\rangle$
is negative definite. Consider
an orthonormal basis $\{e_j\}_{j=1}^n$ of $\mathfrak{g}$ such that
for all $1\leq j,\,k \leq n$, we have
\begin{equation}\label{Killing}
\langle e_j,e_k\rangle =
-\delta_{jk}\;.
\end{equation}

Let $H$ be a complex Hilbert space of finite dimension. Recall that $\|\cdot\|_{op}$
denotes the operator norm on the space $\mathfrak{su}(H)$ of
skew-Hermitian operators. For an operator $A: \mathfrak{su}(H) \to \mathfrak{su}(H)$
we write $|||A|||$ for its operator norm with respect to the operator norm on $\mathfrak{su}(H)$.

\begin{defin}\label{def-1} {\rm
A linear map $t: \mathfrak{g} \to \mathfrak{su}(H)$ is called
a $(\mu,K,\epsilon)$-\emph{almost representation} of
$(\mathfrak{g},\{\cdot,\cdot\})$ if the following assumptions
hold:
\begin{itemize}
\item For all $1\leq j,\,k\leq n$, the \emph{defect}
\begin{equation}\label{alphajkdef}
\alpha_{jk}:=t(\{e_j,e_k\})-[t(e_j),t(e_k)]
\end{equation}
satisfies $\epsilon:=\max_{j,\,k}\|\alpha_{jk}\|_{op}$ ;
\item $K:= \max_j \|t(e_j)\|_{op}$ ;
\item The \emph{almost-Casimir operator}
$\Gamma$ defined by \eqref{aCdef} is invertible with
$\mu:=|||\Gamma^{-1}|||$.
\end{itemize}
}
\end{defin}

%A representation
%$T:(\mathfrak{g},\{\cdot,\cdot\})\to \mathfrak{su}(H)$
%of $(\mathfrak{g},\{\cdot,\cdot\})$ on $H$
%in the usual sense will be called \emph{genuine}.

\begin{thm}\label{thm-main-1}
Let $(\mathfrak{g},\{\cdot,\cdot\})$ be a real
semi-simple compact finite
dimensional Lie algebra. Then for any $c>0$, there
exists a constant $\gamma>0$ with the following property.
Given any $(\mu,K,\epsilon)$-almost representation
$t:\mathfrak{g} \to \mathfrak{su}(H)$ with
$\epsilon \leq \gamma\min(\mu^{-2}K^{-2},\mu^{-1},1)$,
there exists a representation
$\rho: \mathfrak{g} \to \mathfrak{su}(H)$
such that for all $1\leq j\leq n$,
\begin{equation}\label{t-T}
\|t(e_j) - \rho(e_j)\|_{\textup{op}}\leq c\,\mu\,K\epsilon \;.
\end{equation}
\end{thm}

\begin{rem}\label{rem-compar}{\rm
Although more general, this result has a number of drawbacks as compared to Theorem \ref{mainthm-1}, in the case $\mathfrak{g} = \mathfrak{su}(2)$. First, it is unclear to us how to estimate $\mu$ in the case of geometric quantizations of the sphere. Second, even if
we have an {\it ansatz}  $\mu \sim 1$ and $\|x_j\| \sim k \sim \dim H$,
as it should be for an irreducible $k$-dimensional representation,
the existence of a nearby genuine representation is guaranteed only when the defect $\epsilon \lesssim k^{-2}$, as opposed to a less restrictive assumption $\epsilon \lesssim k^{-1}$
provided by Theorem \ref{mainthm-1}.
}
\end{rem}

\medskip\noindent{\sc Discussion on almost irreducibility:}
For representations, the invertibility
of the adjoint Casimir $\Gamma$ is equivalent to irreducibility.
In fact, note that the definition of
almost-Casimir given in \eqref{aCdef} extends to any collection
$X=\{x_1,\dots,x_n\}$ of operators in $\mathfrak{su}(H)$
by the formula
$$\Gamma\sigma:= - \sum_{i=1}^n [[\sigma,x_i],x_i]\;.$$
Furthermore,
\begin{equation}\label{trgamma}
\tr (\Gamma(\sigma)\,\sigma) = \sum_{i=1}^n
\tr\left([\sigma,x_i]^2\right)\;,
\end{equation}
and hence $\Gamma \sigma=0$ if and only if $\sigma$ commutes with all the operators from $X$. In particular,  $\Gamma$ is invertible if and only if the operators
from $X$ possess a common proper invariant subspace.
With this in mind, we are going to compare
$\mu(X):= |||\Gamma^{-1}|||$ with another quantity of geometric flavor which can be interpreted as
a magnitude of irreducibility. Put
$$ d(X):= \min_\Pi \max_j \|(\id-\Pi)\,x_j\,\Pi\|_{op}\;,$$
where $\Pi$ runs over all orthogonal projectors to proper subspaces
$V \subset H$, and $j \in \{1,2,3\}$. Intuitively speaking, smallness of $d$ yields
that the corresponding subspace $V$ is almost invariant.

To this end,  denote by $\cX$ the space of all collections
$X$ whose almost-Casimir is invertible. We say that two positive functions on $\cX$
are {\it equivalent} if their ratio is bounded away from $0$ and $+\infty$
by two constants which depend on $\dim H$.

\medskip
\noindent
\begin{prop}
The functions $\mu^{-1/2}$ and $d$ are equivalent.
\end{prop}

\medskip
\noindent
{\bf Sketch of the proof:}
Denote by $\|A\|_2:= \sqrt{\tr ( A^*A)}$ the Hilbert-Schmidt norm of an operator,
and by $\lambda_1(X)$ the first eigenvalue of $-\Gamma$. The standard inequalities
between the Hilbert-Schmidt norm and the operator norm imply that $\mu^{-1}$
is equivalent to $\lambda_1$.  The claim follows from the inequalities
\begin{equation}\label{eq-1Cas}
d(X)^2 \leq C_1(k)\lambda_1(X)\;,
\end{equation}
and
\begin{equation}\label{eq-2Cas}
\lambda_1(X) \leq C_2(k) d(X)^2\;.
\end{equation}
In order to prove inequality \eqref{eq-1Cas}, take an eigenvector $A$ of $\Gamma$  with $\|A\|_2=1$ corresponding to the first eigenvalue. Since $\tr A =0$,
the spectrum of $A$ can be written as the union of two clusters lying at distance
at least $\sim k^{-2}$ apart. Let $\Pi$ be the spectral projection corresponding to one of them.
Since by \eqref{trgamma} $A$ almost commutes with $x_j$
up to $\epsilon$, one readily deduces
from Lemma \ref{quasimodo} on quasimodes that the image of $\Pi$ is almost invariant under $x_j$.
This yields \eqref{eq-1Cas}.

Inequality \eqref{eq-2Cas} follows from the identity
\begin{equation}\label{eq-identity-Cas}
-(\Gamma(\Pi), \Pi)= 2 \sum_{i=1}^n \big\|[x_i,\Pi]\big\|^2_2\;,
\end{equation}
which holds true for every orthogonal projector $\Pi$.

The details of the argument are left to the reader.
\qed

\medskip

It would be interesting to find sharp bounds on the ratio of $\mu^{-1/2}$ and
$d$ in terms of $\dim H$. At the moment, we cannot compute them even for genuine irreducible representations.

%%%%%%%%%%%%%%%%%%%%%%%%%%%%%%%%%%%%%%%%%%%%%%%%%%%%%%

\medskip
\noindent
{\bf Proof of Theorem \ref{thm-main-1}:}
To simplify the notations, we will often write $x_j:=t(e_j)$
for all $1\leq j\leq n$.
All the estimates in the proof are with respect to the operator
norm of $\mathfrak{su}(H)$ and only depend on
$(\mathfrak{g},\{\cdot,\cdot\})$.

For a linear map $a: \mathfrak{g} \to \mathfrak{su}(H)$,
define an {\it approximate} Elienberg-Chevalley coboundary $d_t a: \mathfrak{g} \otimes \mathfrak{g} \to \mathfrak{su}(H)$ by
$$d_ta (g,h) := [t(g),a(h)] - t(h),a(g)] - a(\{g,h\})\;.$$

%We denote $\bigo(\epsilon_1,\epsilon_2, \dots, \epsilon_k):= O(\epsilon_1\epsilon_2\cdots\epsilon_k)$.

The proof follows the Newton-type iterative process due to Kazhdan \cite{K} adapted to the context of Lie algebras. At the first step we try to find a linear map $a:\mathfrak{g}\to \mathfrak{su}(H)$
so that
\begin{equation}
\overline{t}(g):=t(g)+a(g)
\end{equation}
is a genuine representation.
This yields equation
\begin{equation}\label{eq-homolog-vsp}
\alpha(g,h) - d_t a(g,h) - [a(g),a(h)]=0\;.
\end{equation}
Ignore the third, quadratic in $a$ term, and solve the linearized equation
$d_t a = \alpha$. As we will see, the almost representation $\overline{t}:= t+a$
is closer to a genuine representation. Repeating the process, we get in the limit
the desired genuine representation approximating the original almost representation $t$.

To make this precise, we have to solve the linearized homological equation $d_t a = \alpha$.
This is done by using an effective approximate version of Whitehead's Lemma (see p.88--89 of \cite{J}).

Consider the anti-symmetric $2$-form
$\alpha:\mathfrak{g}\times\mathfrak{g}\to \mathfrak{su}(H)$ defined
for any $g,\,h\in\mathfrak{g}$ by
\begin{equation}
\alpha(g,h):=t(\{g,h\})-[t(g),t(h)]
\end{equation}
and the $1$-form $a:\mathfrak{g}\to \mathfrak{su}(H)$ defined for any
$g\in\mathfrak{g}$ by
\begin{equation}\label{adef}
a(g):=- \sum_{i=1}^n\Gamma^{-1}[\alpha(g,e_i),x_i]\,.
\end{equation}

\begin{lemma}\label{lem-four-vsp}
For all $j,k = 1, \dots n$
\begin{equation}\label{alphajk=da}
\alpha(e_j,e_k) = d_t a(e_j,e_k)+O(\mu^2K^2\epsilon^2)\,.
\end{equation}
\end{lemma}
The lemma is proved at the end of this section.

Let us now consider the linear map $\overline{t}:\mathfrak{g}\to \mathfrak{su}(H)$
defined for all $g\in\mathfrak{g}$ by
\begin{equation}
\overline{t}(g):=t(g)+a(g)\,,
\end{equation}
and set $\overline{x}_j:=\overline{t}(e_j)$ for all $1\leq j\leq n$.
Then  for all $1\leq j\leq n$, by formula \eqref{adef} for $a(e_j)$
we have
\begin{equation}\label{tbarest}
\overline{x}_j\leq K(1+O(\mu\,\epsilon))\,.
\end{equation}
On the other hand, considering for all $1\leq j,\,k\leq n$ the defect
\begin{equation}
\overline{\alpha}_{jk}:=\overline{t}(\{e_j,e_k\})-
[\overline{t}(e_j),\overline{t}(e_k)]\,,
\end{equation}
we see from \eqref{eq-homolog-vsp} that
\begin{equation}\label{alphabarest}
\overline{\alpha}_{jk}= \alpha_{jk} - d_ta(e_j,e_k) - [a(e_j),a(e_k)] =O(\mu^2K^2\epsilon^2)\,.
\end{equation}
Finally, consider the almost-Casimir operator
$\overline{\Gamma}:\mathfrak{su}(H)\to \mathfrak{su}(H)$ defined as in \eqref{aCdef}
with $x_k$ replaced by $\overline{t}(e_k)$ for all $1\leq k\leq n$.
Then we get
\begin{equation*}
\overline{\Gamma}=\Gamma
+\epsilon(1+\mu\epsilon)\,O(\mu K^2)
=\Gamma\Big(\id+\epsilon(1+\mu\epsilon)\,O(\mu^2K^2)\Big)\,.
\end{equation*}
This implies that for any $\delta>0$, there exists a constant
$\gamma>0$ such that
if $\epsilon(1+\mu\epsilon)\leq\gamma/\mu^2K^2$,
then $\overline{\Gamma}$
is invertible and for all $\sigma\in\mathfrak{su}(H)$,
its inverse satisfies
\begin{equation*}
\left\|\overline{\Gamma}^{-1}(\sigma)\right\|_{\textup{op}}
\leq (1+\delta)\mu
\,\|\sigma\|_{op}\,.
\end{equation*}
This, together with the estimates \eqref{tbarest} and
\eqref{alphabarest}, shows that for any $\delta>0$, there
exists $\gamma>0$ such that under the hypothesis
$\epsilon\leq\gamma\min(\mu^{-2}K^{-2},\mu^{-1},1)$,
the linear map
$\overline{t}:\mathfrak{g}\to \mathfrak{su}(H)$ is an
$(\overline{\mu},\overline{K},\overline{\epsilon})$-almost representation with
\begin{equation*}
\overline{\mu}\leq \mu\,(1+\delta)\,,\quad\overline{K}\leq
K\,(1+\delta)
\quad\text{and}\quad\overline{\epsilon}\leq \epsilon\,\delta\,.
\end{equation*}
Taking $\delta>0$ such that $\delta<(1+\delta)^{-4}$, we get that
$\overline{\epsilon}
\leq\gamma\min(\overline{\mu}^{-2}\overline{K}^{-2},\overline{\mu}^{-1},1)$,
and we can reiterate the construction above with the
$(\overline{\mu},\overline{K},\overline{\epsilon})$-almost representation $\overline{t}:\mathfrak{g}\to \mathfrak{su}(H)$ instead of
$t:\mathfrak{g}\to \mathfrak{su}(H)$. At the $N$-th iteration, we get a
$(\mu_N,K_N,\epsilon_N)$-almost representation
$t_N:\mathfrak{g}\to \mathfrak{su}(H)$ with
\begin{equation*}
\mu_N\leq \mu\,(1+\delta)^N\,,\quad K_N\leq
K\,(1+\delta)^N
\quad\text{and}\quad \epsilon_N\leq \epsilon\,\delta^N\,.
\end{equation*}
Writing
$a_N:\mathfrak{g}\to \mathfrak{su}(H)$ for the the $1$-form defined as in
\eqref{adef} for $t_N:\mathfrak{g}\to \mathfrak{su}(H)$, for all
$1\leq j\leq n$ we get
\begin{equation}\label{tnest}
\begin{split}
t_N(e_j)=t_{N-1}(e_j)+a_N(e_j) =t(e_j)+\sum_{k=1}^N a_k(e_j)\\
=t(e_j)+ \sum_{k=1}^N \big((1+\delta)^2\delta\big)^k
O(\mu\,K\epsilon)\,,
\end{split}
\end{equation}
and the sum of the last line converges as $N\to +\infty$
for $\delta>0$ small enough. As $\epsilon_N\to 0$, the
limit map
$\rho:\mathfrak{g}\to \mathfrak{su}(H)$
is a genuine representation,
satisfying the inequality \eqref{t-T} by \eqref{tnest}.
\qed

{\bf Proof of Lemma \ref{lem-four-vsp}:}
First note that by definition, for any $1\leq i,\,j,\,k\leq n$,
we have
\begin{multline*}
0=[\alpha_{jk},x_i]+\alpha(\{e_j,e_k\},e_i)
+[\alpha_{ki},x_j]\\
+\alpha(\{e_k,e_i\},e_j)
+[\alpha_{ij},x_k]+\alpha(\{e_i,e_j\},e_k)\,.
\end{multline*}
Taking the bracket of this identity with $x_i$ and
following the computations of \cite[p.90]{J},
this implies that for all $1\leq j,\,k\leq n$, we have
\begin{multline}\label{Gammalphajk}
\Gamma\alpha_{jk}=
\sum_{i=1}^n
\Big([\alpha(\{e_j,e_k\},e_i),x_i]+ \\
[[\alpha_{ki},x_i],x_j]
-[[\alpha_{ji},x_i],x_k]\Big)-A_{jk}\,,
\end{multline}
with
\begin{multline}\label{A}
A_{jk}:=
-\sum_{i=1}^n\Big([\alpha_{ki},[x_j,x_i]]-
[\alpha(e_k,\{e_i,e_j\}),x_i]\\
-[\alpha_{ji},[x_k,x_i]]
+[\alpha(e_j,\{e_i,e_k\}),x_i]\Big)\,.
\end{multline}
Applying $\Gamma^{-1}:\mathfrak{su}(H)\to \mathfrak{su}(H)$ on both sides of the equality
\eqref{Gammalphajk} and recalling the definition \eqref{adef}
of $a:\mathfrak{g}\to \mathfrak{su}(H)$, we get
\begin{equation}\label{alphajk=da1}
\alpha_{jk}=[x_j,a(e_k)]-[x_k,a(e_j)]-a(\{e_i,e_j\})
-B_{jk}-\Gamma^{-1}A_{jk}\,,
\end{equation}
with
\begin{multline}\label{B}
B_{jk}:=
-\sum_{i=1}^n \Big(\Gamma^{-1}[[\alpha_{ki},x_i],x_j]
-[\Gamma^{-1}[\alpha_{ki},x_i],x_j]\\
-\Gamma^{-1}[[\alpha_{ji},x_i],x_k]
+[\Gamma^{-1}[\alpha_{ji},x_i],x_k]\Big)\,.
\end{multline}
Let us now estimate the terms \eqref{A} and \eqref{B}. First note
that as the Killing form $\langle\cdot,\cdot\rangle$ is $\textup{Ad}$-invariant
and by the explicit
formula \eqref{Killing}, we have
\begin{equation}\label{Killingcomput}
\begin{split}
-\sum_{i=1}^n[\alpha(e_k,\{e_i,e_j\}),x_i]= \sum_{i=1}^n\Big(\sum_{l=1}^n
\langle\{e_i,e_j\},e_l\rangle [\alpha_{kl},x_i]\Big)\\
= \sum_{l=1}^n
\Big[\alpha_{kl},\sum_{i=1}^n
\langle\{e_j,e_l\},e_i\rangle x_i\Big]
=-\sum_{l=1}^n[\alpha_{kl},t(\{e_j,e_l\})]\\
=-\sum_{l=1}^n[\alpha_{kl},[x_j,x_l]]+
\sum_{l=1}^n[\alpha_{kl},\alpha_{jl}]
= -\sum_{l=1}^n[\alpha_{kl},[x_j,x_l]]+O(\epsilon^2)\,.
\end{split}
\end{equation}

Comparing with formula \eqref{A} for $A_{jk}$, this implies
that
\begin{equation}\label{Aest}
\Gamma^{-1}A_{jk}=O(\mu\,\epsilon^2)\,.
\end{equation}
On the other hand, following \cite[p.\,78]{J}
for any $g\in\mathfrak{g}$ and $1\leq j\leq n$,
using the Killing form in the same way than in \eqref{Killingcomput} we get

\begin{multline*}
\Gamma[g,x_j]=-\sum_{i=1}^n[[[g,x_j],x_i],x_i]
=[\Gamma g,x_j]-\sum_{i=1}^n[[g,[x_j,x_i]],x_i]
-\sum_{i=1}^n[[g,x_i],[x_j,x_i]] \\
=[\Gamma g,x_j]-C_{j}(g)-\sum_{i=1}^n[[g,t(\{e_j,e_i\})],x_i]
-\sum_{i=1}^n[[g,x_i],t(\{e_j,e_i\})]
=[\Gamma g,x_j]-C_j(g)
\,,
\end{multline*}
with
$$
C_{j}(g):=-\sum_{i=1}^n[[g,\alpha_{ji}],x_i]\\
-\sum_{i=1}^n[[g,x_i],\alpha_{ji}]\,.
$$
In particular, for any $1\leq i,\,j,\,k,\,l\leq n$,
we have
\begin{multline*}
\Gamma^{-1}[[\alpha_{ij},x_k],x_l]-
[\Gamma^{-1}[\alpha_{ij},x_k],x_l]\\
=\Gamma^{-1}\left([[\alpha_{ij},x_k],x_l]-
\Gamma[\Gamma^{-1}[\alpha_{ij},x_k],x_l]\right)\\
=\Gamma^{-1}C_l(\Gamma^{-1}[\alpha_{ij},x_k])
=O(\mu^2K^2\epsilon^2)\,.
\end{multline*}
Comparing with formula \eqref{B} for $B_{jk}$, we thus get
\begin{equation}\label{Gammaest}
B_{jk}=O(\mu^2K^2\epsilon^2)\,.
\end{equation}
Then via the estimates \eqref{Aest} and \eqref{Gammaest},
the identity \eqref{alphajk=da1} becomes
$$
\alpha_{jk}=[x_j,a(e_k)]-[x_k,a(e_j)]-a(\{e_k,e_j\})
+O(\mu^2K^2\epsilon^2)\,.
$$
This completes the proof of the lemma.
\qed

\section{Equivalence of quantizations}
\label{quantsec}

The basic strategy of the proofs of Theorem
\ref{mainthquant-1} is
to show that geometric
quantizations of the sphere or the torus induce
almost representations of $\mathfrak{su}(2)$ and the quantum
torus respectively, when restricted to a specific set of basic
functions, and then use our Theorems \ref{mainthm-1} and \ref{mainthmT2}.
Let us first start with some generalities on geometric
quantizations of a closed symplecic manifold $(M,\om)$.

\subsection{General setting}

Let $(M,\om)$ be a closed symplectic manifold.
A bi-differential
operator $C: C^\infty(M) \times C^\infty(M) \to C^\infty(M)$
is called a \emph{Hochschild cocycle} if for all
$f_1,\,f_2,\,f_3 \in C^\infty(M)$, we have
\begin{multline}\label{Hochcocycle}
\partial_H C(f_1,f_2.f_3)\\
:= f_1\,C(f_2,f_3)-C(f_1f_2,f_3)+C(f_1,f_2f_3)-
C(f_1,f_2)\,f_3\\
=0\;.
\end{multline}
The operator $\partial_H$ is called the {\it Hochschild differential}.
We will write
$$C_-(f,g) := \frac{C(f,g)-C(g,f)}{2}\quad\text{and}
\quad C_+(f,g) := \frac{C(f,g)+C(g,f)}{2}\;.$$
for the anti-symmetric and symmetric part of $C$.

Assume now that $\{T_k:C^\infty(M)\to\cL(H_k)\}_{k\in\N}$,
satisfy the axioms of Definition \ref{quantdef}.
The associativity of composition
of operators implies that the bi-differential $C_1$
appearing in axiom (P3) is a Hochschild cocycle,
and that for any $f_1,\,f_2,\,f_3\in C^\infty(M)$,
we have
\begin{equation}\label{delHC2}
\partial_H C_2 (f_1,f_2,f_3) = C_1(C_1(f_1,f_2),f_3) - C_1(f_1,C_1(f_2,f_3))\,.
\end{equation}
Furthemore, the axiom (P2) is equivalent to the fact that
\begin{equation}\label{c1-}
C_1^-(f,g)=\frac{i}{2}\{f,g\}\,,
\end{equation}
for all $f,\,g\in C^\infty(M)$.
Then formula \eqref{Hochcocycle} for $C_1^-$ is a consequence
of the Leibniz rule for the Poisson bracket, and
this shows that
$C_1^+$ is a symmetric Hochschild cocycle.
Then
by \cite[Th.\,2.15]{RG},
it is a \emph{Hochschild coboundary}, meaning that there exists
a differential operator $D:C^\infty(M)\to C^\infty(M)$
vanishing on constants such that
for $f,\,g\in C^\infty(M)$, we have
\begin{equation}\label{adef-1}
C_1^+(f,g)=D(f)g+fD(g)-D(fg)\,.
\end{equation}
Furthermore, the axiom (P1) implies that the operator
$T_k(f)\in\End(H_k)$ is Hermitian for all $k\in\N$ big enough
if and only if $f\in C^\infty(M,\C)$ is real valued.
As the square of a Hermitian operator is Hermitian,
the axiom (P3) then shows that
$C^+_1$ is a real-valued bi-differential operator,
so that $D$ has real coefficients.

Let us now assume that $C_1^+\equiv 0$, and consider
the bi-differential operators $\hat{C}_1$ and $\hat{C}_2$
defined by interchanging $f,\,g\in C^\infty(M)$ in axiom (P3)
in the following way, as $k\to+\infty$,
\begin{equation}
T_k(g)T_k(g) =: T_k\left(fg+\frac{1}{k}
\hat{C}_1(f,g)+\frac{1}{k^2}\hat{C}_2(f,g)\right) + \bigo(1/k^3)\;.
\end{equation}
Then we have $\hat{C}_1(f,g)= C_1(g,f)=-C_1(f,g)$ and $\hat{C}_2(f,g)=C_2(g,f)$.
On the other hand, associativity of composition of operators implies
that \eqref{delHC2} holds for $\hat{C}_1$ and $\hat{C}_2$,
one readily checks that $\partial_H C_2 =\partial_H \hat{C}_2$.
Therefore, $C_2^-=(C_2-\hat{C}_2)/2$ is an anti-symmetric Hochschild cocycle,
and by \cite[Th.\,2.15]{RG},
there exists a $2$-form $\alpha\in\Om^2(M,\C)$ so that
for all $f,\,g\in C^\infty(M)$, we have
\begin{equation}\label{tilc2-}
C_2^-(f,g)=\frac{i}{2}
\alpha(\text{sgrad}\,f,\text{sgrad}\,g)\,.
\end{equation}
Furthermore, by axiom (P1) as above and the fact that the commutator
of Hermitian operators is skew-Hermitian,
the axiom (P3) implies that
the bi-differential operator $iC_2^-$ is real valued,
so that $\alpha$ is a real $2$-form.

The proofs of Theorem \ref{mainthquant-1} and \ref{trconj}
are based on a natural operation
on quantizations, which we call a \emph{change of variable}.
Specifically, given a geometric quantization
$\{T_k:C^\infty(M)\to\cL(H_k)\}_{k\in\N}$, and a differential
operator $D:C^\infty(M)\to C^\infty(M)$, set
\begin{equation}\label{chgeofvar}
T_k^D(f):=T_k\left(f+\frac{1}{k}D\,f\right)\,,
\end{equation}
for all $f\in C^\infty(M)$ and all $k\in\N$. Then one readily
checks that the maps $\{T_k^D:C^\infty(M)\to\cL(H_k)\}_{k\in\N}$,
satisfy the axioms of Definition \ref{quantdef},
and that for any
$f\in C^\infty(M)$, we have the estimate
\begin{equation}\label{chgofvareq}
\left\|T_k(f)-T^D_k(f)\right\|_{op}=\bigo(1/k)\,,
\end{equation}
as $k\to+\infty$.
We will write $C_{1,D}$ and $C_{2,D}$ for the associated
bi-differential operators of axiom (P3).

We will use the operation of change of variables
to reduce the proof of Theorem
\ref{mainthquant-1} to a class of remarkable
quantizations, described by the following result.

\begin{lemma}\label{finequant}
Assume that $(M,\om)$ satisfies $\dim M=2$. Then for any geometric
quantization $\{T_k:C^\infty(M)\to\cL(H_k)\}_{k\in\N}$, there
exists a differential operator
$D:C^\infty(M)\to C^\infty(M)$ vanishing on constants
such that the
bi-differential operators of axiom (P3)
associated with the induced quantization
$\{T_k^D:C^\infty(M)\to\cL(H_k)\}_{k\in\N}$,
satisfy
\begin{equation}\label{finefla}
C_{1,D}^+(f,g)=0\quad\quad\text{and}\quad\quad
C_{2,D}^-(f,g)=-\frac{i}{2}c\,\{f,g\}\,,
\end{equation}
for all $f,\,g\in C^\infty(M)$, where $c\in\R$ is constant.
\end{lemma}
\begin{proof}
One readily computes that a change of variable \eqref{chgeofvar}
associated to a differential operator
$D:C^\infty(M)\to C^\infty(M)$ acts on the
bi-differential operators $C_1^+$ and $C_2^-$ via the following
formula, for all $f,\,g\in C^\infty(M)$,
\begin{equation}\label{cjA}
\begin{split}
C_{1,D}^+(f,g)&=C_1^+(f,g)+D(f)g+fD(g)-D(fg)\,,\\
C_{2,D}^-(f,g)&=C_2^-(f,g)+\frac{i}{2}\Big(\{D(f),g\}
+\{f,D(g)\}-D(\{f,g\})\Big)\,.
\end{split}
\end{equation}
In particular, formula \eqref{adef-1} shows that
there is an operator $D$ satisfying
$C_{1,D}^+\equiv 0$, determined up to the addition of
a derivation $\delta:C^\infty(M)\to C^\infty(M)$.

Let now $D:C^\infty(M)\to C^\infty(M)$ be such that
that $C_{1,D}^+\equiv 0$, and let
$\alpha_D\in\Om^2(M,\R)$ be the two form of formula
\eqref{tilc2-} associated with $C^-_{2,D}$.
Recall that we assume $\dim M=2$, so that $H^2(M,\R)$
is $1$-dimensional, generated by the cohomology
class $[\om]$. Then if we set
\begin{equation}
c:=\frac{1}{2\pi}\int_{M}\,\alpha_D\,,
\end{equation}
we know that there exists
a $1$-form $\theta\in\Om^1(M,\R)$ such that
\begin{equation}\label{alphaA}
\alpha_D=c\,\om+d\theta\,.
\end{equation}
On the other hand,
for all $f,\,g\in C^\infty(M)$,
we have by definition
\begin{multline}\label{dtheta}
d\theta(\sgrad f,\sgrad g)\\
=\theta(\sgrad\{f,g\})
-\{\theta(\sgrad f),g\}
-\{f,\theta(\sgrad g)\}\,.
\end{multline}
Then if we consider the derivation $\delta:C^\infty(M)\to C^\infty(M)$
defined for all $f\in C^\infty(M)$ by $\delta f:=\theta(\sgrad f)$, formulas \eqref{cjA} and \eqref{alphaA} imply
\begin{equation}
C_{2,D+\delta}^-(f,g)=\frac{i}{2}c\,\om(\sgrad f,\sgrad g)=-
\frac{i}{2}c\,\{f,g\}\,,
\end{equation}
and $C_{1,D+\delta}^+=C_{1,D}^+\equiv 0$.
This shows the result.
\end{proof}

Let us end this Section with an existence Theorem, which was
already alluded to in Example \ref{exam-BT-desc}.

\begin{thm}\label{BMS}
\cite{BMS}
Let $(M,\om)$ be a closed symplectic manifold
with $[\om]\in 2\pi H^2(M,\Z)$
admitting a complex structure compatible with $\om$.
Then there exists a geometric quantization
$\{T_k:C^\infty(M)\to\cL(H_k)\}_{k\in\N}$, such that
for all $f\in C^\infty(\T^2,\C)$, its $\C$-linear extension
satisfies
\begin{equation}
\|T_k(f)\|_{op}\leq\|f\|_{\infty}\,.
\end{equation}
\end{thm}

\subsection{Proof of Theorem \ref{mainthquant-1}}
\label{proofquantsec}

Using the estimate \eqref{chgofvareq} and
Lemma \ref{finequant}, we see that it suffices to establish
Theorem \ref{mainthquant-1} for geometric
quantizations for which there
is a constant $c\in\R$ such that
$C_1^+\equiv 0$ and $C_2^-=-\frac{i}{2}c\,\{\cdot,\cdot\}$.
All geometric quantizations considered in this Section will
thus satisfy this property.

The proof of Theorem \ref{mainthquant-1} for the two cases $M=S^2$ and
$M=\T^2$ follows the same strategy: we first
establish \eqref{quanteq} for a finite set of functions generating
a dense subalgebra of $C^\infty(M)$, and then use the quasi-multiplicativity
axiom (P3) in a careful way to extend it to the whole $C^\infty(M)$ by density.
\newline

\medskip\noindent {\sc Case of $M=S^2$:}  We will use the Cartesian coordinate functions
$u_1,\,u_2,\,u_3\in C^\infty(S^2)$ of $S^2$
seen as the unit sphere in $\R^2$. The induced volume form
$\om$ is the standard volume form of volume $2\pi$,
and these coordinate functions satisfy the commutation
relation
\begin{equation}\label{ujcomm}
\{u_j,u_{j+1}\} = -2u_{j+2}\,,
\end{equation}
for all $j\in\Z/3\Z$.
Then given a quantization $\{T_k:C^\infty(S^2)\to\cL(H_k)\}_{k\in\N}$ with
$C_1^+\equiv 0$ and $C_2^-=-\frac{i}{2}c\,\{\cdot,\cdot\}$,
one readily checks from Definition
\ref{quantdef} and the commutation relations \eqref{ujcomm}
that the assumptions of
Theorem \ref{mainthm-1} are satisfied for
the constant $c\in\R$ and
the operators
$x_1,\,x_2,\,x_3\in\mathfrak{su}(H_k)$ defined for all $k\in\N$
and $j\in\Z/3\Z$ by
\begin{equation}
x_j:=\frac{ik}{2}\,\frac{k}{k-c}\,T_k(u_j)\,,
\end{equation}
where $u_1,\,u_2,\,u_3\in C^\infty(S^2)$ are the Cartesian coordinates
of $S^2\subset\R^3$. As the assumption
$\limsup_{k\to+\infty}\dim H_k/k<2$ implies in particular that
$\dim H_k<2(k+c)$ for all $k\in\N$,
it follows that $c\in\Z$ and that $\dim H_k=k+c$
for all $k\in\N$ big enough, which proves the first statement
\eqref{dimquant}.

Furthermore, Theorem \ref{mainthm-1} implies that
there exist operators
$X_1,\,X_2,\,X_3\in\mathfrak{su}(H_k)$ generating an irreducible
representation of $\mathfrak{su}(2)$
such that for all $1\leq j\leq 3$,
\begin{equation}\label{Tkalmrep}
\left\|\frac{ik}{2}\,\frac{k}{k-c}\,T_k(u_j)-X_j\right\|_{op}=\bigo(1)\,.
\end{equation}
Now if $\{S_k:C^\infty(S^2)\to\cL(H_k)\}_{k\in\N}$
is another quantization with same sequence of Hilbert spaces,
we get in the same way operators
$\til{X}_1,\,\til{X}_2,\,\til{X}_3\in\mathfrak{su}(H_k)$
generating an irreducible
representation of $\mathfrak{su}(2)$
such that for all $1\leq j\leq 3$,
\begin{equation}\label{Tkalmrep'}
\left\|\frac{ik}{2}\,\frac{k}{k-c}\,S_k(u_j)-\til{X}_j
\right\|_{op}=\bigo(1)\,.
\end{equation}
As any two
irreducible representations of $\mathfrak{su}(2)$
with same dimension are isomorphic, formulas \eqref{Tkalmrep}
and \eqref{Tkalmrep'} show that there exist
unitary operators $U_k:H_k\to H_k$ for all $k\in\N$
such that for all $j\in\Z/3\Z$,
\begin{equation}\label{TkalmeqQk}
\|T_k(u_j)-U_k^{-1}S_k(u_j)U_k\|_{op}=\bigo(1/k)\,.
\end{equation}
Set $Q_k:=U_k^{-1}S_kU_k$ for all $k\in\N$,
and note that by transitivity, it suffices
to establish \eqref{quanteq} when $\{T_k:C^\infty(S^2)\to\cL(H_k)\}_{k\in\N}$
is the Berezin-Toeplitz quantization of Theorem \ref{BMS}.

%%\todo{detailed proof of the gap added}
Consider the decomposition
of $L^2(S^2,\C)$ into the direct sum of eigenspaces $H_n$
of the Laplace-Beltrami operator
$\Delta$ with eigenvalue $2n(n+1)$, for each $n\in\N$.
Using for instance \cite[Cor.\,1.1]{BinXu}, we know
that for any $N\in\N$, there exists $C_N>0$ such that
for any $n\in\N^*$ and $f\in H_n$, we have
\begin{equation}\label{Weyllaw}
\|f\|_{C^N}\leq C_N n^{2N} \|f\|_{L^2}\,.
\end{equation}
Recall on the other hand that for any $n\in\N$, the eigenspace
$H_n$ is isomorphic to the irreducible $SO(3)$-representations of
highest weight $n\in\N$
with respect to $S^1\subset SO(3)$ rotating along the $u_3$-axis,
and write $f_n\in H_n$ for the unit highest weight vectors.
Via the identification with \emph{spherical harmonics} and
following e.g. \cite[Ex.\,15.4.1,\,\S\,15.5]{Arfken},
we have the following recursion formula in $n\in\N$,
\begin{equation}\label{recfla}
\begin{split}
f_{n+1}&=\sqrt{\frac{2n+3}{2n+2}}f_1\,f_n\,,\\
f_1&=-(u_1+iu_2)\,.
\end{split}
\end{equation}
Let us prove by induction that there exists constants
$\alpha>0$ and $M\in\N$ such that for any $n\in\N$
and all $k\in\N$, we have
\begin{equation}\label{TkalmeqQkpol}
\|T_k(f_n)-Q_k(f_n)\|_{op}
\leq \alpha\frac{n^M}{k}\,,
\end{equation}
The case $n=1$ readily follows from \eqref{TkalmeqQk} and
formula \eqref{recfla} for $f_1$. On the other hand,
axioms (P1), (P3) and the estimate \eqref{Weyllaw} give
constants $\alpha_0>0$ and $N\in\N$ such that
for any $n\in \N$ and $k\in\N$, we have
\begin{equation}\label{TkalmeqQkpolest}
\begin{split}
\|Q_k(f_1\,f_n)-Q_k(f_1)Q_k(f_n)\|_{op}
&\leq\frac{\alpha_0}{k}\|f_1\|_{C^N}n^{2N}\,,\\
\|Q_k(f_n)\|_{op}&\leq \alpha_0\,n^{2N}\,,
\end{split}
\end{equation}
and the same holds for $\{T_k:C^\infty(S^2)\to\cL(H_k)\}_{k\in\N}$.
Let now $n\in\N$ be such that \eqref{TkalmeqQkpol} and holds,
and recall  by assumption that $\|T_k(f)\|_{op}\leq\|f\|_{\infty}$ for all
$f\in C^\infty(S^2,\C)$ and $k\in\N$.
Then using the sub-multiplicativity of the
operator norm, we get that for any $f\in C^\infty(S^2)$,
\begin{equation}\label{algcomput}
\begin{split}
\|T_k&(f_1\,f_n)-Q_k(f_1\,f_n)\|_{op}\\
&\leq\|T_k(f_1)T_k(f_n)-Q_k(f_1)Q_k(f_n)\|_{op}+
2\frac{\alpha_0}{k}\|f_1\|_{C^N}n^{2N}\\
&\leq\|T_k(f_1)(T_k(f_n)-Q_k(f_n))\|_{op}\\
&\quad\quad+\|(T_k(f_1)-Q_k(f_1))Q_k(f_n)\|_{op}
+2\frac{\alpha_0}{k}\|f_1\|_{C^N}n^{2N}\\
&\leq
\frac{\alpha}{k}\|f_1\|_{\infty}n^M+\frac{C\alpha_0}{k}n^{2N}
+2\frac{\alpha_0}{k}\|f_1\|_{C^N}n^{2N}\,,
\end{split}
\end{equation}
where $C>0$ comes from the estimate \eqref{TkalmeqQk} and
formula \eqref{recfla} for $f_1$.
As $\|f_1\|_\infty=\max_{x\in S^2}\,|u_1+iu_2|=1$, we can
choose $\alpha=\alpha_0(C+2\|f_1\|_{C^N})$
and $M=2N+1$ in \eqref{TkalmeqQkpol} to get
\begin{equation}
\begin{split}
\frac{\alpha}{k}\sqrt{\frac{2n+3}{2n+2}}(n^M+n^{2N})
&\leq\frac{\alpha}{k}(n+1)(n^{M-1}+n^{2N-1})\\
&\leq\frac{\alpha}{k}(n+1)^M\,,
\end{split}
\end{equation}
where we used the fact that $n\sqrt{\frac{2n+3}{2n+2}}\leq n+1$ for all $n\in\N$.
Using \eqref{recfla}, this implies \eqref{TkalmeqQkpol}
with $n$ replaced by $n+1$, and
thus for all $n\in\N$ by induction.

Let us now establish the estimate \eqref{TkalmeqQkpol} for
all functions in $H_n$, for each $n\in\N$.
First, by definition of the action of $SO(3)$ on the unit sphere $S^2$,
we see that \eqref{TkalmeqQk} implies the existence of a constant $C>0$
such that for any $g\in SO(3)$, any $j\in\Z/3\Z$ and all $k\in\N$, we have
\begin{equation}
\|T_k(g^*u_j)-Q_k(g^*u_j)\|_{op}\leq C/k\,.
\end{equation}
Note on the other hand that for any $g\in SO(3)$, the functions
$g^*f_n\in C^\infty(S^2)$, $n\in\N$, are again highest weight vectors
with respect to $S^1\subset SO(3)$ rotating along the $g^*u_3$-axis.
We can then repeat the reasoning above replacing
$f_n$ by $g^*f_n$ for all $n\in\N$ to get
\begin{equation}\label{Ginvh+}
\|T_k(g^*f_n)-Q_k(g^*f_n)\|_{op}
\leq \alpha\frac{n^M}{k}\,,
\end{equation}
for any $g\in SO(3)$, with same constants $\alpha>0$ and $M\in\N$.
Recall on the other hand that
the standard volume form $\om$ on $S^2$ is the pushforward
of the Haar measure on $SO(3)$.
Then following e.g. \cite[III.3.3.a]{Busch},
for any $f\in H_n$ we can consider
its \emph{coherent state decomposition}
\begin{equation}
f=\frac{n+1}{2\pi}\,\int_{S^2}\,\langle f,g^*f_n\rangle_{L^2}~g^*f_n~\om_{[g]}\,,
\end{equation}
where $g\in SO(3)$ is any representative of
$[g]\in S^2\simeq SO(3)/S^1$. Then using \eqref{Ginvh+} and
Cauchy-Schwartz inequality, we get
\begin{equation}\label{finalest}
\begin{split}
\|T_k(f)&-Q_k(f)\|_{op}\\
&\leq
\frac{n+1}{2\pi}\,\int_{S^2}\,\left|\langle f,g^*f_n\rangle_{L^2}\right|~
\|T_k(g^*f_n)-Q_k(g^*f_n)\|_{op}~\om_{[g]}\\
&\leq \alpha\|f\|_{L^2} \frac{(n+1)n^{M}}{k}\,.
\end{split}
\end{equation}
Take now any $f \in C^\infty(S^2)$, and consider its spectral
decomposition into the eigenspaces of $\Delta$, so that
$f = \sum_{n\in\N} a_n\,\varphi_n$,
with $\varphi_n\in H_n$ and $\|\varphi_n\|_{L^2}=1$
for all $n\in\N$.
Since $f$ is smooth, the sequence $(a_n)_{n\in\N}$
decays faster than any power of $n$,
so that using \eqref{finalest}, there exists $C'>0$ such that
\begin{equation}\label{fourierlegendre}
\|T_k(f)-Q_k(f)\|_{op} \leq \alpha\sum_{n\in\N}
\frac{(n+1)n^{M}}{k}a_n \leq \frac{C'}{k}\;.
\end{equation}
This shows formula \eqref{quanteq} in the case of $M=S^2$.
\newline

\medskip\noindent {\sc Case of $M=\T^2$:}
Write $(q_1,q_2)\in\T^2=\R^2/\Z^2$ for the standard coordinates,
so that the standard volume form
of volume $2\pi$ writes $\om=2\pi\,dq_1\wedge dq_2$.
We will use the functions $u_1,\,u_2\in C^\infty(\T^2,\C)$
defined for any $q:=(q_1,q_2)\in\T^2\simeq \R^2/\Z^2$ and
$j=1,\,2$ by
\begin{equation}\label{uj=expi}
u_j(q):=e^{2\pi i q_j}\,,
\end{equation}
which satisfy the commutation relation
\begin{equation}\label{Poisu1u2}
\{u_1,u_2\}=2\pi u_1u_2\,.
\end{equation}
Following for instance in \cite[\S\,2]{BFFLS},
we consider the \emph{Moyal-Weyl star product} over
$(C^\infty(\T^2),\{\cdot,\cdot\})$, defined as
in \eqref{DQintro} with coefficients
$\til{C}_1$ and $\til{C_2}$ satisfying
$\til{C}_1^+=\til{C}_2^-=0$ and for all $f,\,g\in C^\infty(\T^2)$,
\begin{equation}\label{MWC2}
\til{C}_2^+(f,g)=-\frac{1}{32\pi^2}
\left(\frac{\partial^2}{\partial q_1^2}f\,
\frac{\partial^2}{\partial q_2^2}g
-2\frac{\partial^2}{\partial q_1\partial q_2}f\,
\frac{\partial^2}{\partial q_1\partial q_2}g
+\frac{\partial^2}{\partial q_2^2}f\,
\frac{\partial^2}{\partial q_1^2}g\right)\,.
\end{equation}
Then given a geometric quantization
$\{T_k:C^\infty(\T^2)\to\cL(H_k)\}_{k\in\N}$ with
$C_1^+\equiv 0$ and $C_2^-=-\frac{i}{2}c\,\{\cdot,\cdot\}$,
we get that $C_1=\til{C}_1$ via formula \eqref{c1-},
and using \eqref{delHC2}, we get that $C_2-\til{C}_2$
is a Hochschild cocycle \eqref{Hochcocycle}.
On the other hand, we have
$C_2^--\til{C}_2^-=-\frac{i}{2}c\,\{\cdot,\cdot\}$,
which also satisfies \eqref{Hochcocycle},
and we thus get $C_2^+-\til{C}_2^+$ is
a symmetric Hochschild cocycle, hence a coboundary by
\cite[Th.\,2.15]{RG}. As in formula \eqref{adef},
this means that
there exists a differential
operator $D_2:C^\infty(\T^2)\to C^\infty(\T^2)$ vanishing on
constants such that
\begin{equation}\label{tilC2+}
\til{C}_2^+(f,g)=C_2^+(f,g)+D_2(f)g+fD_2(g)-D_2(fg)\,.
\end{equation}
Consider the following
change of variables at second order in $1/k$,
for all $f\in C^\infty(\T^2)$,
\begin{equation}\label{TD2ndorder}
T_k^{D_2}(f):=T_k\left(f+\frac{1}{k^2}{D_2}(f)\right)\,.
\end{equation}
This again defines a geometric quantization
in the sense of Definition \ref{quantdef},
with associated bi-differential operators
$C_{1,D_2}$ and $C_{2,D_2}$ of axiom (P3)
satisfying
$C_{1,D_2}^+\equiv 0$ and $C_{2,D_2}^-=-\frac{i}{2}c\,\{\cdot,\cdot\}$,
while formula \eqref{tilC2+} implies
$C_{2,D_2}^+=\til{C}_2^+$.
Note also that $T_k^{D_2}(u_j)^*=T_k^{D_2}(u_j^{-1})$ for each $j=1,\,2$
by Definition
\ref{quantdef} and formula \eqref{uj=expi}.
Then using formula \eqref{MWC2}, we get as $k\to+\infty$,
\begin{equation}
\begin{split}
T_k^{D_2}(u_j)T_k^{D_2}(u_j)^*&=\id+\frac{i}{2}\left(\frac{1}{k}
-\frac{c}{k^2}\right)
T_k^{D_2}(\{u_j,u_j^{-1}\})+\bigo(1/k^3)\\
&=\id+\bigo(1/k^3)
\end{split}\,,
\end{equation}
and
\begin{multline}
T_k^{D_2}(u_1)T_k^{D_2}(u_2)=
T_k^{D_2}(u_1u_2)+\frac{i}{2}\left(\frac{1}{k}
-\frac{c}{k^2}\right)T_k^{D_2}(\{u_1,u_2\})\\
-\frac{1}{32\pi^2k^2}T_k^{D_2}\left(\frac{\partial^2}
{\partial q_1^2}u_1
\frac{\partial^2}{\partial q_2^2}u_2\right)+O(1/k^3)\\
=T_k^{D_2}(u_1u_2)+\frac{i}{2}
\frac{2\pi}{k+c}T_k^{D_2}(u_1u_2)
-\frac{(2\pi)^2}{8k^2}T_k^{D_2}\left(u_1u_2\right)+O(1/k^3)\\
=e^{2\pi i/2(k+c)}T_k^{D_2}(u_1u_2)+O(1/k^3)\,,
\end{multline}
while in the same way,
\begin{equation}
T_k^{D_2}(u_2)T_k^{D_2}(u_1)=e^{-2\pi i/2(k+c)}T_k^{D_2}(u_1u_2)+O(1/k^3)\,.
\end{equation}
We then
see that the operators $x_j:=T_k^{D_2}(u_j)$ for all $j=1,\,2$
and $k\in\N$, satisfy the assumptions of
Theorem \ref{mainthmT2} for
the constant $c\in\R$ as above.
As the assumption
$\limsup_{k\to+\infty}\dim H_k/k<2$ implies in particular that
$\dim H_k<2(k+c)$ for all $k\in\N$,
it follows that $c\in\Z$ and that $\dim H_k=k+c$
for all $k\in\N$ big enough, which proves  \eqref{dimquant}.
Furthemore, Theorem \ref{mainthmT2} and formula \eqref{TD2ndorder}
imply that
there exist unitary operators
$X_1,\,X_2\in\End(H_k)$ satisfying
$X_1X_2=e^{2\pi i/(k+c)}X_2X_1$ and not
preserving any non-trivial proper subspace,
such that
\begin{equation}\label{TkalmrepT2}
\|T_k(u_j)-X_j\|_{op}=\bigo(1/k)\quad\text{for all}\quad
j=1,\,2\,.
\end{equation}

Note that the explicit formula \eqref{X1X2}
shows that for any two such pairs $X_1,\,X_2\in\End(H_k)$
and $\til{X}_1,\,\til{X}_2\in\End(H_k)$, there
exists a unitary operator $U:H_k\to H_k$ and
$p:=(p_1,p_2)\in\T^2\simeq \R^2/\Z^2$ such that
for each $j=1,\,2$, we have
\begin{equation}\label{tilX=X}
\til{X}_j=e^{2\pi i p_j}U^{-1}X_jU\,.
\end{equation}
Setting $m_j:=\lfloor (k+c) p_j\rfloor\in\N$ for each $j=1,\,2$,
and considering the unitary operator
$U_{m_1,m_2}:=X_1^{-m_2}X_2^{m_1}\in\End(H_k)$,
we get
the following estimates in operator norm as $k\to+\infty$,
for all $p=(p_1,p_2)\in\T^2$ and all unitary operators
$X_1,\,X_2\in\End(H_k)$ satisfying
the above commutation relations,
\begin{equation}
U_{m_1,m_2}X_jU_{m_1,m_2}^{-1}=e^{-2\pi i m_j/(k+c)}X_j=
e^{-2\pi i p_j}X_j+\bigo(1/k)\,.
\end{equation}
Thus for any two such pairs of sequences $X_1,\,X_2\in\End(H_k)$,
$k\in\N$,
and $\til{X}_1,\,\til{X}_2\in\End(H_k)$,
$k\in\N$, we get
a sequence of unitary operators
$U_k:H_k\to H_k$, $k\in\N$, such that
\begin{equation}\label{tilX=Xk}
\til{X}_j=U^{-1}_kX_jU_k+\bigo(1/k)\,.
\end{equation}
%On the other hand, if $\tau_p:\T^2\to\T^2$ is the translation
%operator by $p\in\T^2$, then we have
%\begin{equation}
%\tau_p^*u_j=e^{2\pi ip_j}u_j\,.
%\end{equation}
Then if we have two quantizations
$\{T_k,\,Q_k:C^\infty(\T^2)\to\cL(H_k)\}_{k\in\N}$
with same sequence of Hilbert spaces satisfying
$\dim H_k=k+c$ for all $k\in\N$ big enough,
they both satisfy \eqref{UVdef} for two different
pairs $X_1,\,X_2\in\End(H_k)$
and $\til{X}_1,\,\til{X}_2\in\End(H_k)$,
and formula \eqref{tilX=Xk} shows that there
exists unitary operators $U_k:H_k\to H_k$ for all
$k\in\N$ such that for each $j=1,\,2$,
\begin{equation}\label{quanteqT2uj}
\|U_k^{-1}Q_k(u_j)U_k-T_k(u_j)\|_{op}=\bigo(1/k)\,.
\end{equation}
Now by transitivity as above, it suffices
to establish \eqref{quanteq} when $\{T_k:C^\infty(\T^2)\to\cL(H_k)\}_{k\in\N}$,
is the Berezin-Toeplitz quantization of Theorem \ref{BMS}. Then
by a straightforward adaptation of the computation
\eqref{algcomput} with $f_1$ replaced by $u_1,\,u_2$ respectively
and $f_n$ replaced by $u_1^nu_2^m$,
we get by induction on $n,\,m\in\Z$ that there exist
constants $\alpha>0$ and $M\in\N$, depending only on the quantizations,
such that
%%\todo{argument adapted following the case of $S^2$}
\begin{equation}\label{TkalmeqQkpolT2}
\|T_k(u_1^nu_2^m)-Q_k(u_1^nu_2^m)\|_{op}
\leq \alpha\frac{(|n|+|m|)^M}{k}\,.
\end{equation}
Now for any $f \in C^\infty(\T^2)$, consider its Fourier
expansion
\begin{equation}
f = \sum_{m,\,n\in\Z} a_{m,n}u_1^nu_2^m\,.
\end{equation}
Since $f$ is smooth, the coefficients $a_{n,m}$, $n,\,m\in\Z$,
decay faster than any polynomial in $|n|,\,|m|$.
Using the estimate \eqref{TkalmeqQkpolT2} in the same way as with
\eqref{fourierlegendre}, this shows formula \eqref{quanteq} in the case
of $M=\T^2$, and concludes the proof of Theorem \ref{mainthquant-1}.

\subsection{Traces of quantizations}
\label{trsec}

Note that in the previous section, we showed
in particular that for geometric quantizations of
$M=S^2$ or $\T^2$, the constant
$c\in\R$ appearing in Lemma \ref{finequant} is an integer,
uniquely determined by the condition $\dim H_k=k+c$
for all $k\in\N$ big enough. This fact can be refined
for geometric
quantizations satisfying the following additional axiom.

\begin{defin}\label{trdef}
{\rm A geometric quantization $\{T_k:C^\infty(M)\to\cL(H_k)\}_{k\in\N}$
of a closed symplectic manifold
$(M,\om)$ of dimension $\dim M=2d$
is said to satisfy the {\it trace axiom} if
there exists a function $R\in C^\infty(S^2)$ such that
for all $f\in C^\infty(S^2)$, we have
\begin{equation}\label{trfla}
\tr\;T_k(f) = \left(\frac{k}{2\pi}\right)^d\,
\int_{M} f\,R_k\,\frac{\omega^d}{d!}\;,
\end{equation}
for a sequence of functions $R_k\in C^\infty(M)$
satisfying the following estimate as $k\to+\infty$,
$$R_k=1+\frac{1}{k}\,R+\bigo(1/k^2)\,.$$
}
\end{defin}

We then have the following refinement of Lemma \ref{finequant},
relating this trace with the coefficient $C_2^-$.

\begin{thm}\label{trconj}
Let $M=S^2$ or $\T^2$ be endowed with the
standard volume form $\om$ of volume $2\pi$.
Then if
$\{T_k:C^\infty(M)\to\cL(H_k)\}_{k\in\N}$
is a geometric quantization
with $C_1^+\equiv 0$
satisfying the trace axiom of Definition \ref{trdef},
we have for all $f,\,g\in C^\infty(M)$,
\begin{equation}\label{c2R}
C_2^-(f,g)=-\frac{i}{2}R\,\{f,g\}\,.
\end{equation}
\end{thm}
%
%The operator $D$ appearing in formula \eqref{cjfla}
%has been extensively studied in \cite{IKP30}.
%In the case of K\"{a}hler quantizations, it
%has been computed to be $D=-\Delta/4$, where $\Delta$ is the
%Laplace-Beltrami operator of $S^2$ endowed with the corresponding
%K\"{a}hler metric.
%Then the second of the
%identities \eqref{cjfla} for
\begin{proof}
Let $T_k:C^\infty(M)\to\cL(H_k)$, $k\in\N$, be a
geometric quantization with
$C_1^+\equiv 0$ satisfying the trace axiom of
Definition \eqref{trdef},
and recall the form $\alpha\in\Om^2(M,\R)$ of
formula \eqref{tilc2-}. Let $c\in\R$ and $\theta\in\Om^1(M,\R)$ be such that $\alpha=c\,\om+d\theta$,
as in formula \eqref{alphaA}, and write
\begin{equation}\label{Rtheta}
d\theta=:R_\theta\,\om\,,
\end{equation}
with $R_\theta\in C^\infty(M)$. Considering the
change of variable \eqref{chgeofvar} induced by
the derivation $\delta:C^\infty(M)\to C^\infty(M)$
defined by $\delta f:=\theta(\sgrad f)$, we compute
\begin{equation}\label{deltaf=Rtheta}
\int_{M}\,\delta f\,\om =-\int_{M}\,f\,d\theta =
-\int_{M}\,R_\theta\,f\,\om\;.
\end{equation}
Then one readily computes that the quantization
$\{T_k^\delta:C^\infty(M)\to\cL(H_k)\}_{k\in\N}$ induced by $\delta$
as in \eqref{chgeofvar}
also satisfies the trace axiom of Definition
\ref{trdef}, where the
function $R\in C^\infty(M)$ is replaced by the function
$R_\delta:=R-R_\theta$. On the other hand, we know from
the proof of Lemma \ref{finequant} that
$C_{1,\delta}^+=C_1^+\equiv 0$ and
$C_{2,\delta}^-=-\frac{i}{2}c\,\{\cdot,\cdot\}$, and
from the proof of Theorem
\ref{mainthquant-1} above that $c\in\R$ is an integer
satisfying $\dim H_k=k+c$ for all
$k\in\N$ big enough.
Applying formula \eqref{trfla} to $f=1$ and using that $T_k^\delta(1)=\id$,
we get
\begin{equation}\label{RAmean}
\frac{1}{2\pi}\int_{M}\,R_\delta\,\om=c\,.
\end{equation}
On the other hand, using the axioms (P2) and (P3),
we get for any $f,\,g\in C^\infty(M)$ that
as $k\to+\infty$,
\begin{equation}
\begin{split}
i\left(1-\frac{c}{k}\right)\tr\,T_k^\delta(\{f,g\})&=k\,\tr\left([T_k^\delta(f),T_k^\delta(g)]+\bigo(1/k^3)\right)\\
&=\bigo(1/k)\,.
\end{split}
\end{equation}
Now as every function with zero mean can be written as a sum
of Poisson brackets (see e.g. \cite[Theorem 1.4.3]{Banyaga}), we get that
\begin{equation}
\int_{M}\,f\,\om=0\quad\text{implies}\quad\tr\,T^\delta_k(f)=\bigo(1/k)
~\text{as}~k\to+\infty\,.
\end{equation}
Using formula \eqref{trfla} again, we see that this is
possible if and only $R_\delta$ is constant, equal to $c\in\Z$
by formula \eqref{RAmean}. We thus have $R=c+R_\theta$,
and by formulas \eqref{dtheta} and
\eqref{Rtheta}, we get
\begin{equation}
C_2^-(f,g)=-\frac{i}{2}c\,\{f,g\}-\frac{i}{2}R_\theta\,\{f,g\}
=-\frac{i}{2}R\,\{f,g\}\,.
\end{equation}
This gives the result.
\end{proof}

Together with Theorem \ref{mainthmT2}, Theorem
\ref{trconj}
implies the following extension of Theorem \ref{mainthquant-1}
in a special case. Denote by  $\tau_{p}:\T^2\to\T^2$ the translation by an element $p\in\T^2$.

\begin{thm}\label{3/2eqth}
Let $\{Q_k,\,T_k:C^\infty(\T^2)\to\cL(H_k)\}_{k\in\N}$ be two geometric
quantizations of the torus satisfying the trace axiom of Definition
\ref{trdef},
with the same sequence of Hilbert spaces, same $C_1^+$, and same
$R\in C^\infty(\T^2)$. Assume furthermore that the function $R$ is constant,
and that the bi-differential operator $C_1^+$ is translation invariant.
Then there exist a sequence $\{U_k:H_k\to H_k\}_{k\in\N}$
of unitary operators and a sequence $\{p_k\in\T^2\}_{k\in\N}$ of points
in $\T^2$ such that for any $f\in C^\infty(\T^2)$, we have as $k\to+\infty$,
\begin{equation}\label{3/2eqfla}
\|U_k^{-1}Q_k(\tau_{p_k}^*f)U_k-T_k(f)\|_{op}=\bigo(1/k^{3/2})\,.
\end{equation}
\end{thm}
\begin{proof}
First note that as $C_1^+$ is translation invariant, the differential operator
$D:C^\infty(\T^2)\to C^\infty(\T^2)$ appearing in formula
\eqref{adef-1} can be chosen to be translation invariant
as well. Consider the common change of variable defined for all
$f\in C^\infty(\T^2)$ and $k\in\N$ by
\begin{equation}\label{QkDTkD}
Q_k^D(f):=Q_k\left(f+\frac{1}{k}D(f)\right)
\quad\text{and}\quad T_k^D(f):=T_k\left(f+\frac{1}{k}D(f)\right)\,,
\end{equation}
so that $C_{1,D}^+\equiv 0$. Since both $D$ and $\om$ are translation invariant, and $D$ vanishes on constants,
we have
\begin{equation}
\int_{\T^2}\,D(f)\,\om =0\;,
\end{equation}
for all $f\in C^\infty(\T^2)$.
We then see that
the quantizations \eqref{QkDTkD} also
satisfy the trace axiom of Definition \ref{trdef}
with function $R\in C^\infty(\T^2)$ unchanged, and by Theorem
\ref{trconj}, they have same $C_{2,D}^-$, given
for all $f,\,g\in C^\infty(\T^2)$
by the formula
\begin{equation}
C_{2,D}^-(f,g)=-\frac{i}{2}R\,\{f,g\}\,.
\end{equation}
Furthermore,
the trace axiom of Definition \ref{trdef} implies
in particular that $\limsup_{k\to+\infty}\dim H_k/k<2$.
As $R$ is constant by assumption,
we can then follow the proof of Theorem
\ref{mainthquant-1} in Section \ref{proofquantsec}
in the torus case $\T^2$, replacing the quantizations
$\{Q_k,\,T_k:C^\infty(\T^2)\to\cL(H_k)\}_{k\in\N}$
by the quantizations
$\{Q_k^D,\,T_k^D:C^\infty(\T^2)\to\cL(H_k)\}_{k\in\N}$
constructed above.
Using the full strength of Theorem \ref{mainthmT2},
we get unitary operators
$X_1,\,X_2\in\End(H_k)$ satisfying
$X_1X_2=e^{2\pi i/(k+c)}X_2X_1$
such that the following analogue of
formula \eqref{TkalmrepT2} holds,
\begin{equation}\label{Tkalmrep3/2}
\|T_k^D(u_j)-X_j\|_{op}=\bigo(1/k^{3/2})\quad\text{for all}\quad
j=1,\,2\,.
\end{equation}
Furthermore, the same holds for $Q_k^D$ with operators
$\til{X}_1,\,\til{X}_2\in\End(H_k)$ such that
\begin{equation}\label{tilX=X3/2}
\til{X}_j=e^{2\pi i p_j}U^{-1}X_jU\quad\text{for all}\quad
j=1,\,2\,,
\end{equation}
for a unitary operator $U:H_k\to H_k$ and
$p:=(p_1,p_2)\in\T^2\simeq \R^2/\Z^2$.
On the other hand, note that by definition \eqref{uj=expi}
of $u_j\in C^\infty(\T^2,\C)$ for all $j=1,\,2$,
if $\tau_p:\T^2\to\T^2$ is the translation
operator by $p\in\T^2$, then we have
\begin{equation}
\tau_p^*u_j=e^{2\pi ip_j}u_j\,.
\end{equation}
Using now the commutation relation \eqref{Poisu1u2} and the fact that
$C_{1,D}^+\equiv 0$, we get from axiom (P3) a constant $\alpha>0$ and
a constant $N\in\N$ such that for any $m,\,n\in\Z$, we have
%%\todo{argument precised}
\begin{equation}
\begin{split}
\left\|T_k^D(u_1)T_k^D(u_1^nu_2^m)-\left(1+ \frac{2\pi in}{k}\right)
T_k^D(u_1^{n+1}u_2^m)\right\|_{op}
&\leq\frac{\alpha}{k^2}(|n|+|m|)^{N}\,,\\
\left\|T_k^D(u_2)T_k^D(u_1^nu_2^m)-\left(1+ \frac{2\pi i m}{k}\right)
T_k^D(u_1^{n}u_2^{m+1})\right\|_{op}
&\leq\frac{\alpha}{k^2}(|n|+|m|)^{N}\,.
\end{split}
\end{equation}
Using this estimate and through a traightforward adaptation
of the proof in Section \ref{proofquantsec} for $M=\T^2$,
we then get a sequence of unitary
operators $\{U_k:H_k\to H_k\}_{k\in\N}$,
$k\in\N$ and a sequence of points $\{p_k\in\T^2\}_{k\in\N}$,
such that for any $f\in C^\infty(S^2)$,
we have the following estimate as $k\to+\infty$,
\begin{equation}\label{quanteqTkD3/2}
\|U_k^{-1}Q_k^D(\tau_{p_k}^*f)U_k-T_k^D(f)\|_{op}
=\bigo(1/k^{3/2})\,,
\end{equation}
where $\tau_{p}:\T^2\to\T^2$ denotes the translation by $p\in\T^2$. As the common change of variable \eqref{QkDTkD}
is invariant by translation, this readily implies the result.
\end{proof}

Theorem \ref{trconj} is of specific interest in the theory of
\emph{deformation
quantization} of the Poisson algebra
$(C^\infty(M),\{\cdot,\cdot\})$. To see this,
consider the following
extension of axiom (P3).

\begin{defin}\label{totexp}{\rm
A geometric quantization $\{T_k:C^\infty(M)\to\cL(H_k)\}_{k\in\N}$
of a closed symplectic manifold
$(M,\om)$ is said to satisfy the \emph{star product axiom} if
there exists a collection of bi-differential operators
$C_j$, $j\in\N$, such that for all $m\in\N$
and all $f,\,g\in C^\infty(M)$,
\begin{equation}\label{totexpfla}
T_k(f)T_k(g) = T_k\left(fg+\sum_{j=1}^{m-1}
\frac{1}{k^j}\,C_j(f,g)\right)+ \bigo(1/k^{m})\;.
\end{equation}
}
\end{defin}

The name for this axiom is justified by the fact that,
together with the other axioms of
Definition \ref{quantdef}, this induces a
\emph{differential star product} $*$ on the ring
of formal power series
$C^\infty(M,\C)[[\hbar]]$, with formal parameter $\hbar$.
Specifically, the formula
\begin{equation}\label{stardef}
f*g:= fg+\sum_{j=1}^\infty\hbar^j C_j(f,g)\,,
\end{equation}
for all $f,\,g\in C^\infty(M)$, defines
an associative unital $\C[[\hbar]]$-linear
product $*$ on $C^\infty(M,\C)[[\hbar]]$ satisfying
$f*g-g*f= i\hbar\{f,g\}+\bigo(\hbar^2)$.
Setting $\hbar=1/k$, we see that \eqref{totexp}
reads formally as the {\it star product axiom}
\begin{equation} \label{totexp-1}
T_k(f)T_k(g)= T_k(f*g)\;,
\end{equation}
where this equality is understood as an asymptotic expansion
with respect to the operator norm.

Working with formal power series in $\hbar$, one
can extend the notions
\eqref{chgeofvar} and \eqref{TD2ndorder}
of a change of variable over any subset $U\subset M$
as a map
$A: C^\infty(U,\C)[[\hbar]]\longrightarrow
C^\infty(U,\C)[[\hbar]]$ satisfying $A(1)=1$ and
\begin{equation}\label{starchgeofvar}
A(f):=f + \sum_{j=1}^{+\infty} \hbar^j D_j\,f\,,
\end{equation}
for all compactly supported $f\in C^\infty(U)$,
where $D_j$ are differential operators
for all $j\in\N$. This acts on a star product
$*$ via the formula
\begin{equation}\label{stareq}
f*_A g:= A^{-1}(A(f) * A(g))\,,
\end{equation}
where $*_A$ is defined on compactly supported functions
$f,\,g\in C^\infty(U,\C)$. In the theory of
deformation quantization, this is also called a
\emph{star-equivalence}.
For change of variables of the form $A(f)=f+\hbar D\,f$
for any $f\in C^\infty(M)$,
one readily checks that $*_A$ is
the star product
\eqref{totexp-1} associated to the geometric quantization
$\{T_k^D:C^\infty(M)\to\cL(H_k)\}_{k\in\N}$ of \eqref{chgeofvar}.

Following
\cite[\S\,1,\,p.229]{NT1} (see also \cite[\S\,2,\,p.220]{Kar}),
one can define the
\emph{canonical trace} of a differential
star product $\star$ over a closed symplectic manifold
$(M,\om)$ of dimension $\dim M=2d$ as the map
$\tr_\hbar:C^\infty(M)[[\hbar]]\to\C[[\hbar]]$ such that
for any $f\in C^\infty(M)$
supported over a contractible Darboux chart $U\subset M$,
we have
\begin{equation}\label{cantrdef}
\tr_\hbar(f)=(2\pi\hbar)^{-d}\int_X\,A_U(f)\,\frac{\om^d}{d!}\,,
\end{equation}
where $A_U:C^\infty(U)[[\hbar]]\to C^\infty(U)[[\hbar]]$ is
a change of variable making $\star$ equal to the usual
Moyal-Weyl star product over $\R^{2d}$ in these
Darboux charts. We will
not need the full definition of the Moyal-Weyl star product,
but only that it satisfies $C_1^+=C_2^-=0$.
The following result is then a
consequence of Theorem \ref{trconj}.

\begin{cor}\label{thm-trace}
Let $M=S^2$ or $\T^2$ be endowed with the standard volume form $\om$
of the total area $2\pi$. Let
$\{T_k:C^\infty(M)\to\cL(H_k)\}_{k\in\N}$ be a
geometric quantization
satisfying the trace axiom of Definition \ref{trdef} and the
star product axiom of Definition \ref{totexp}.
Then for all $f\in C^\infty(S^2)$, we have the
asymptotic expansion
$$ \tr\,T_k(f)=\tr_{\hbar}(f)  + \bigo(1/k)\;,$$
as $k = 1/\hbar \to +\infty$.
\end{cor}
\begin{proof}
Take $f\in C^\infty(M)$ to be compactly supported in a Darboux
chart $U\subset M$, and let
$A_U:C^\infty(U)[[\hbar]]\to C^\infty(U)[[\hbar]]$ be a local
change of variable making the induced star product \eqref{stardef}
equal to the Moyal-Weyl star
product. Let us write
\begin{equation}
A_U(f)=f+\hbar\,D_U\,f+O(\hbar^2)\,,
\end{equation}
and write $\til{C}_1$ and $\til{C}_2$ for the bi-differential
operators of \eqref{stardef} associated with the star
product $*_{A_U}$ over $U$.
Note that terms of order $\hbar^2$ and more do not affect
$\til{C}_1^+$ and $\til{C}_2^-$, and by formula
\eqref{cjA},
the condition
$\til{C}_1^+\equiv 0$
determines $D_U:C^\infty(U)\to C^\infty(U)$
up to a derivation.
In particular,
by the trace axiom \eqref{trfla}, one sees
that both the usual trace and the canonical
trace change the same way under a change of variable of the
form \eqref{chgeofvar}. By Lemma \ref{finequant},
it suffices to show the result for quantizations
which already satisfy $C_1^+\equiv 0$.

Let then $\{T_k:C^\infty(M)\to\cL(H_k)\}_{k\in\N}$ be
a geometric quantization with $C_1^+\equiv 0$
and satisfying the trace axiom \eqref{trfla},
so that
we are under the hypotheses of Theorem \ref{trconj}.
Then by formula \eqref{cjA}, the condition
$\til{C}_1^+\equiv 0$
implies that $D_U:C^\infty(U)\to C^\infty(U)$
has to be a derivation in that case.
Furthermore,
formulas \eqref{cjA} and \eqref{dtheta}
show that in order to also have
$\til{C}_2^-\equiv 0$, this derivation has to
be of the form $D_U\,f:=-\theta(\sgrad\,f)$ for all
compactly supported
$f\in C^\infty(U)$, where $\theta\in\Om^1(M,\R)$
satisfies
\begin{equation}\label{c2-exact}
C_2^-(f,g)=\frac{i}{2}d\theta(\textup{sgrad} f,\textup{sgrad} g)\,,
\end{equation}
for all compactly supported $f,\,g\in C^\infty(U)$.
Note that this is compatible with formula \eqref{tilc2-},
as all $2$-forms over a contractible open set $U\subset M$
are exact.
Then by definition \eqref{cantrdef}
of the canonical trace, for all $f\in C^\infty(M)$ with compact
support in $U\subset X$, we then have
\begin{equation}\label{eq-trace-final}
\begin{split}
\tr_\hbar(f)&=\frac{1}{2\pi\hbar}\int_X\,
\left(f-\hbar
\,\theta(\textup{sgrad}f)\right)\,\om
+O(\hbar)\\
&=\frac{1}{2\pi\hbar}\int_X\,\left(f+\hbar\,
fR_U\right)\,\om
+O(\hbar)\,,
\end{split}
\end{equation}
where $R_U\in C^\infty(U)$ is defined by the formula
\begin{equation}
d\theta=:R_U\,\om|_U\,.
\end{equation}
Therefore, by formula \eqref{c2-exact},
for all compactly supported
$f,\,g\in C^\infty(U)$, we have
\begin{equation}
C_2^-(f,g)=-\frac{i}{2}\,R_U\,\{f,g\}\,.
\end{equation}

By Theorem \ref{trconj}, the trace $\tr\,T_k(f)$ is given by the
last term in formula \eqref{eq-trace-final}, and hence coincides with the
canonical trace $\tr_\hbar(f)$ up to $\bigo(1/k)$. This completes
the proof of the corollary.
\end{proof}

Corollary \ref{thm-trace} naturally
leads to the following conjecture.

\begin{conjecture}\label{conjtrace} Let
$\{T_k:C^\infty(M)\to\cL(H_k)\}_{k\in\N}$ be a
geometric quantization of a closed symplectic manifold $(M,\om)$
satisfying the trace axiom of Definition \ref{trdef} and
the star product axiom of Definition \ref{totexp}.
Then for all $f\in C^\infty(M)$ and $m\in\N$,
we have the
asymptotic expansion
$$ \tr\,T_k(f)=\tr_{\hbar}(f)  + O(1/k^m) \;,$$
as $k = 1/\hbar \to+\infty$.
\end{conjecture}
The trace axiom of Definition \ref{trdef}
is a basic property of
Berezin-Toeplitz quantizations of closed K\"{a}hler manifolds,
and the fact that these quantizations
satisfy the star product axiom of Definition
\ref{totexp} has been
shown by Schlichenmaier in \cite{Sch}.
Then Conjecture \ref{conjtrace} for Berezin-Toeplitz
quantizations of closed K\"{a}hler manifolds
has been established by Hawkins in \cite[Cor.\,10.5]{Haw00}.

\begin{rem}\label{starrmk}
{\rm
As explained for instance in \cite[\S\,6]{RG}, there exists
a notion of characteristic class for differential star-products $*$
over symplectic manifolds, which has been
introduced by Deligne in \cite{Del} as an element $c(\star)$
of the affine space $\hbar^{-1}[\om]+H^2(M,\R)[[\hbar]]$.
By the work of Fedosov \cite{F2} and Nest and Tsygan \cite{NT1,NT2},
this class is known to classify star-products up to
\emph{star-equivalence} \eqref{stareq}.
Then we have the relation
\begin{equation}\label{Delclass}
c(\star)=\hbar^{-1}[\om]+c\,[\om]+\bigo(\hbar)\,,
\end{equation}
where $c\in\R$ is the constant produced from
$T_k:C^\infty(M)\to\cL(H_k)$, $k\in\N$
by Lemma \ref{finequant}. Then for geometric
quantizations satisfying star product axiom \eqref{totexp-1},
the proof of Theorem \ref{mainthquant-1}
computes this constant to be an integer via the formula $\dim H_k=k+c$ for all $k\in\N$.
The Deligne-Fedosov class
of the standard Berezin-Toeplitz quantizations of closed K\"{a}hler manifolds
has been computed by Hawkins in \cite[Th.\,10.6]{Haw00}
and Karabegov and Schlichenmaier in \cite{KS01}.}
\end{rem}

%%%%%%%%%%%%%%%%%%%%%%%%%%%%%%%%%%%%%%%%%%%%%%%%%%%%%%%%%%%%%

\medskip\noindent{\bf Acknowledgement.}
%%\todo{acknowledgement to Ood added}
L.P. thanks University of Chicago, where
a part of this paper was written, for hospitality and an excellent research atmosphere.
We thank D.Treschev for a useful comment, and O. Shabtai for an attentive
reading of the manuscript and pointing out a number of mistakes.

\Addresses

\end{document}